\documentclass[sigconf]{acmart}

\usepackage{hyperref}
\usepackage{url}
\usepackage{microtype}
\usepackage{xcolor}
\usepackage{xspace}

\usepackage{bm}
\usepackage{amsfonts}
\usepackage{bbm}

\usepackage{amssymb}

\usepackage[utf8]{inputenc}
\usepackage[T1]{fontenc}
\usepackage[normalem]{ulem}

\usepackage{booktabs}
\usepackage{multirow}
\usepackage{subcaption}
\usepackage{float}

\usepackage{enumitem}
\setlist[itemize]{align=parleft,left=1em..2em}
\setlist[enumerate]{align=parleft,left=1em..3em}

\usepackage{amsmath}
\usepackage{mathtools}
\usepackage{algorithm2e}

\RestyleAlgo{ruled}
\SetKwComment{Comment}{$\triangleright$\ }{}
\makeatletter 
\g@addto@macro{\@algocf@init}{\SetKwInOut{Parameter}{Parameters}} 
\makeatother

\newtheorem{definition}{Definition}
\newtheorem{lemma}{Lemma}
\newtheorem{remark}{Remark}
\newtheorem{theorem}{Theorem}

\newcommand{\sign}{sign}

\newcommand{\exponent}{\mathsf{e}}
\newcommand{\significand}{\mathsf{d}}
\newcommand{\signbit}{\mathsf{b}}

\newcommand{\ie}[0]{i.e.,\xspace}
\newcommand{\eg}[0]{e.g.,\xspace}
\newcommand{\etal}[0]{et al.\xspace}

\newcommand{\ReLU}{\mathsf{ReLU}}
\newcommand{\clip}{\mathsf{clip}}
\newcommand{\mina}{\mathsf{min}}
\newcommand{\maxa}{\mathsf{max}}

\newcommand{\PDAW}{\mathsf{ReluPGD}}
\newcommand{\Perturb}{\mathsf{LinApproxPerturbDir}}

\newcommand{\Neighbor}{\mathsf{Neighbor}}
\newcommand{\leftp}{\mathsf{beforeFP}}
\newcommand{\rightp}{\mathsf{afterFP}}
\newcommand{\step}{s}
\newcommand{\sample}{\mathsf{sample}}

\newcommand{\neighb}{\mathsf{neighb}}
\newcommand{\candidates}{\mathsf{candidates}}

\AtBeginDocument{%
}

\copyrightyear{2024}
\acmYear{2024}
\setcopyright{rightsretained}
\acmConference[AISec '24]{Proceedings of the 2024 Workshop on Artificial Intelligence and Security }{October 14--18, 2024}{Salt Lake City, UT, USA}
\acmBooktitle{Proceedings of the 2024 Workshop on Artificial Intelligence and Security (AISec '24), October 14--18, 2024, Salt Lake City, UT, USA}
\acmDOI{10.1145/3689932.3694761}
\acmISBN{979-8-4007-1228-9/24/10}

\makeatletter
\gdef\@copyrightpermission{
 \begin{minipage}{0.3\columnwidth}
   \href{https://creativecommons.org/licenses/by-nd/4.0/}{\includegraphics[width=0.90\textwidth]{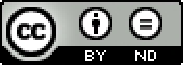}}
 \end{minipage}\hfill
 \begin{minipage}{0.7\columnwidth}
   \href{https://creativecommons.org/licenses/by-nd/4.0/}{This work is licensed under a Creative Commons Attribution-NoDerivs International 4.0 License.}
 \end{minipage}
 \vspace{5pt}
}
\makeatother

\begin{document}

\title[Getting a-Round Guarantees: Floating-Point Attacks on Certified Robustness]{Getting a-Round Guarantees: \\ Floating-Point Attacks on Certified Robustness}

\author{Jiankai Jin}
\orcid{0009-0009-1008-482X}
\affiliation{%
  \institution{The University of Melbourne}
  \city{Melbourne}
  \country{Australia}}
\email{jiankaij@student.unimelb.edu.au}

\author{Olga Ohrimenko}
\orcid{0000-0002-9735-0538}
\affiliation{%
  \institution{The University of Melbourne}
  \city{Melbourne}
  \country{Australia}}
\email{oohrimenko@unimelb.edu.au}
\authornote{Authors made equal contributions to this work}

\author{Benjamin~I.~P. Rubinstein}
\orcid{0000-0002-2947-6980}
\affiliation{%
  \institution{The University of Melbourne}
  \city{Melbourne}
  \country{Australia}}
\email{brubinstein@unimelb.edu.au}
\authornotemark[1]





\begin{abstract}
Adversarial examples pose a security risk as they can alter decisions of 
a machine learning classifier through slight input perturbations.
Certified robustness has been proposed as a mitigation where given an input $\mathbf{x}$, 
a classifier returns a prediction and a certified radius $R$
with a provable guarantee that any perturbation to $\mathbf{x}$ with $R$-bounded norm will not alter the classifier's prediction.
In this work, we show that these guarantees can be invalidated 
due to limitations of floating-point representation that cause rounding errors.
We design a rounding search method that can efficiently exploit 
this vulnerability to find adversarial examples against state-of-the-art certifications
in two threat models, that differ in how the norm of the perturbation is computed.
We show that the attack can be carried out against linear classifiers that have exact certifiable guarantees 
and against neural networks that have conservative certifications.
In the weak threat model, our experiments demonstrate attack success rates over 50\% on random linear classifiers, 
up to 23\% on the MNIST dataset for linear SVM, and up to 15\% for a neural network.
In the strong threat model, the success rates are lower but positive.
The floating-point errors exploited by our attacks can range from small to large (e.g., $10^{-13}$ to $10^{3}$) 
--- showing that even negligible errors can be systematically exploited to invalidate guarantees provided by certified robustness.
Finally, we propose a formal mitigation approach based on rounded interval arithmetic, 
encouraging future implementations of robustness certificates to account for 
limitations of modern computing architecture to provide sound certifiable guarantees.
\end{abstract}

\ccsdesc[500]{Security and privacy~Formal methods and theory of security}
\ccsdesc[500]{Computing methodologies~Machine learning}

\keywords{Certified Robustness, Floating-Point Errors, Adversarial Examples}

\maketitle

\section{Introduction}

Robustness of modern image classifiers has come under scrutiny due to a plethora of results demonstrating adversarial examples---small
perturbations to benign inputs that cause models to mispredict, even when such perturbations are not evident to the human eye~\cite{madry2017towards,carlini2017towards,42503,43405}.
If a learned model is used in critical applications such as self-driving cars, clinical settings or malware detection, such easily added perturbations can have severe consequences.
As a result, research focus has shifted to training models robust to adversarial perturbations, that come endowed with \textit{certified robustness}.

Mechanisms for providing {robustness certification} aim to bound a model $f$'s {sensitivity to} a certain level of perturbation.
At a high level, such mechanisms return a radius $R$ around a test input $\mathbf{x}$ 
with a guarantee that for any $\mathbf{x}'$ within $R$ distance from $\mathbf{x}$, $f(\mathbf{x}) = f(\mathbf{x}')$.
How $R$ is computed, whether it is sound and/or complete depends on the mechanism.
For example, bound propagation~\cite{crown,beta-crown} transfers the upper and lower bounds from the output layer to the input layer of a neural network,
and gives a lower bound on the perturbation needed to flip the classification.

Given the extensive research on certified robustness, can such mechanisms protect against adversarial examples in practice?
In this paper, we show that the limits posed by floating-point arithmetic invalidate guarantees of several prominent mechanisms and their implementations.
Despite proofs of robustness guarantees, they mostly assume real numbers can be represented exactly.
Unfortunately, this critical (implicit) assumption cannot hold on computers with finite number representations.
Since floating-point (FP) numbers can represent only a subset of real values, rounding is likely to occur
when computing robust guarantees and can cause overestimation of the certified radius~$R$. 
Thus, adversarial examples may exist within the computed radius despite claims of certification.

We devise a rounding search method that can efficiently discover such adversarial examples 
in two threat models, that differ in how the norm of the perturbation is computed.
Our method is inspired by the traditional adversarial example search methods 
such as PGD~\cite{madry2017towards} and C\&W~\cite{carlini2017towards}.
However, we find that such existing methods do not effectively exploit 
the rounding of a certified radius as the search space they explore is large 
(\ie the number of examples to check becomes intractable due to the large number of floating-point values) 
and instances of inappropriate rounding do not necessarily follow model gradients.
To this end, our method is different from these search methods in two aspects:
(1) instead of relying on back propagation, it leverages the piecewise linear property of {ReLU networks to find coarse-level} perturbation directions; 
(2) it then searches in a much finer scale by sampling floating-point neighbors of a potential adversarial example.
The first aspect allows us to narrow down the search space closer to the certified radius 
and efficiently find adversarial examples.
The second aspect enables our search method to find adversarial examples with perturbation norms 
that are just smaller than the certified radius (\eg in the 13th decimal place), which PGD and C\&W cannot find.
{Compared to other works that find robustness violations~\cite{jia2021exploiting,zombori2020fooling}, our attack
method is arguably stronger as it works on unmodified target models with unaltered instances as opposed to specially crafted models or instances.}
We discuss the potential impact of our attacks on robustness guarantees in Section~\ref{app:limit_disclosure}.

One's first intuition to mitigate the overestimation of certified radii exploited by the above attacks might be to adopt slightly more conservative radii
(\eg using $R-\gamma$ for some positive constant $\gamma$, \eg $\gamma=1.0$). Unfortunately, such radii are not in general sound and choosing $\gamma$ is inherently error prone. 
That is, we show that the amount of overestimation can depend on the data (\eg number of features) 
and model (\eg number of operations) and that attacking $R-1.0$ is still possible.
To this end, we propose a defense based on rounded interval arithmetic that has soundness guarantees and can be easily integrated into mechanisms for computing certified radii.
In summary our contributions are:
\begin{itemize}
\item We explore a class of attacks that invalidate the implementations of certified robustness 
(\ie find adversarial examples within the certified radius). 
Our attacks exploit rounding errors due to limited floating-point representation of real numbers.
\item We devise a rounding search method that systematically discovers such adversarial examples under two threat models.
The weak model assumes that attacks need only have floating-point norms that violate certifications 
(\eg in the case where the norm is computed using common software libraries). 
The strong model makes no such assumption: the true (real-valued) norm of attacks must violate certifications 
(\eg in the case where the library that computes the square root for the norm can represent a real value or its range).
\item We show that our attacks work against exact certifications of linear models~\cite{cohen2019certified}, and
against a conservative certified radius returned by a prominent neural network verifier~\cite{beta-crown} on a network.
Our attack success rate differs between learners and threat models. 
In the weak threat model, our success rates are over 50\% for random linear classifiers and 15\% on an MNIST neural network model.
In the strong threat model, the attack success rates are lower but are still non-zero.
For all cases, in theory, the certification should guarantee a 0\% success rate
for such attacks within certified radii.
\item We propose a defense based on rounded interval arithmetic, with strong theoretical 
and empirical support for mitigating rounding search attacks.
\end{itemize}

\section{Background and preliminaries}
\label{sec:bg}
Let input instance $\mathbf{x} = (x_1, x_2, \ldots, x_D)$ 
be a vector in $\mathbb{R}^D$ with $x_i$ denoting the $i$th component of $\mathbf{x}$.
We adopt the $\ell_2$ norm $\|\mathbf{x}\|_2 = \left( \sum_{i=1}^D |x_i|^2\right)^{1/2}$, 
written $\|\mathbf{x}\|$ where the norm is understood from context, 
when measuring the length of instance vectors, and its induced metric $\|\mathbf{x} - \mathbf{z}\|$
for distances. We also use the $\ell_1$ norm $\|\mathbf{w}\|_1=\sum_{i=1}^D |w_{i}|$.
We consider the task of learning a classifier $f$ mapping an instance in $\mathbb{R}^D$
to a label in $\{-1,1\}$ or to a $K$-class label in $[K]=\{1,\ldots,K\}$. 
Given a training set of $N$ examples $(\mathbf{x}^{(j)}, y^{(j)}) \in \mathbb{R}^D\times \{-1,1\}$ drawn i.i.d. from some unknown
distribution $P$, the goal of a learner $\mathcal{A}$ is to output some classifier $f$ with low risk $\mathbb{E}_P[\ell(f(\mathbf{X}),Y)]$ for some loss
$\ell(\cdot)$ of interest.

\paragraph{Linear models.}
We consider learners $\mathcal{A}$ over two-class linear classifiers of the form $f(\mathbf{x})=\sign(\mathbf{w}^T\mathbf{x}+b)$, 
for parameters $\mathbf{w}\in\mathbb{R}^D$ the model weights and $b\in\mathbb{R}$ model bias. Specifically we consider as learners:
(1) uniformly random sampling of $(\mathbf{w},b)$ within a closed, bounded set, without consideration of training data;
(2) the linear support vector machine that minimizes the $\ell_1$-regularized squared hinge loss 
$\sum_{i=1}^n (1-y^{(i)} (\mathbf{w}^T \mathbf{x}^{(i)} + b))^2_+ + \lambda\|\mathbf{w}\|_1$.
Classifiers resulting from learners (1) and (2) directly make predictions as $f(\mathbf{x})=\sign(\mathbf{w}^T\mathbf{x}+b)$.

\paragraph{Neural networks.}
We consider fully-connected feed-forward neural networks, which are known to be universal approximators~\cite{hornik1991approximation}. 
Such a network forms output scores by composing a sequence of layers $F_i(\cdot)$ as: 
$F(\mathbf{x}) = \left(F_n \circ F_{n-1} \circ \cdots \circ F_1\right)(\mathbf{x})$. 
The neural network forms its classification within classes $[K]=\{1,\ldots,K\}$ by simple majority vote $f(x)=\arg\max_{k\in [K]} (F(\mathbf{x}))_k$. 
Layers might take the form $F_i(\mathbf{x})=\sigma(\boldsymbol{\theta}_i^T \mathbf{x}+\hat{\boldsymbol{\theta}}_i)$ 
for some weights~$\boldsymbol{\theta}_i$ and vector of model biases~$\hat{\boldsymbol{\theta}}_i$.
For activation function $\sigma$ we consider either linear $\sigma(x)=x$ 
or rectified linear units (ReLUs) $\sigma(x)=\ReLU(x)=\max\{0,x\}$. 
Gradients of neural network outputs can be computed with respect to weights (for training) 
or inputs (for finding adversarial examples---see next) using backpropagation as typically implemented by reverse accumulation in autodiff libraries.

\paragraph{Adversarial examples.}
Given an input instance $\mathbf{x}$, a classifier $f$, and a target label $t \neq f(\mathbf{x})$, 
$\mathbf{x}'$ is a \textit{targeted} adversarial example~\cite{42503} if $f(\mathbf{x}')=t$
where $\mathbf{x}'$ is reachable from $\mathbf{x}$ according to some chosen threat model. 
In the vision domain, it is common to assume that small $\ell_p$ perturbations to $\mathbf{x}$ will go unnoticed by human observers. 
In this paper we consider $\ell_2$ distance, \ie $\|{\mathbf{x}-\mathbf{x}'}\| \le \Delta$ for some small perturbation limit~$\Delta$. 
An adversarial example in the multi-class setting is \textit{untargeted} if $t$ is not specified, 
however, it is possible to find an $\mathbf{x}'$ such that $f(\mathbf{x}') \neq f(\mathbf{x})$ while $\mathbf{x}'$ and $\mathbf{x}$ are still close.
Appropriate methods for finding adversarial examples depend on the model under attack. 
For a two-class linear classifier $f$, 
its weight vector $\mathbf{w}$ is normal to the decision boundary $\mathbf{w}^T\mathbf{x}+b=0$. 
Therefore, there always exists a closest adversarial example $\mathbf{x}'$ just beyond the decision boundary in the direction $\mathbf{w}$.
See Figure~\ref{fig:direction} for an illustration. 
For neural networks, and non-linear classifiers in general, the process is more involved.
Two popular white-box approaches are due to Carlini and Wagner (C\&W)~\cite{carlini2017towards}
and Madry~\etal called Projected Gradient Descent (PGD)~\cite{madry2017towards}.
These methods assume white-box access to the models (that is, access to architectures and model weights)
and hence an attacker can run as many queries as they want and observe intermediate gradients on inputs of their choice.
At a high level, PGD and C\&W view a search for an adversarial example as an optimization problem where one tries to find~$\mathbf{x}'$
with a flipped label while minimizing the distance to $\mathbf{x}$,
and differ in how the perturbation norm $\Delta$ is updated during the search.

\paragraph{Floating-point representation.}
Floating-point values represent reals using three binary numbers:
a sign bit $\signbit$, an exponent $\exponent$,  and a significand $d_1d_2\ldots d_{\significand}$.
For example, 64-bit (double precision) floating-point numbers allocate 1 bit for $\signbit$, 11 bits for $\exponent$,
and 52 bits for the significand. Such a floating-point number is defined to be 
$(-1)^\signbit \times (1.d_1d_2\ldots d_{\significand})_2  \times 2^{\exponent-1023}$.
Floating points can represent only a finite number of real
values. Hence, computations involving floating-point numbers often need to be rounded up or down
to their nearest floating-point representation, as specified by the IEEE 754 Standard for Floating-Point Arithmetic~\cite{IEEE-FP}.

\subsection{Certified robustness}
\label{sec:bqcert}
A robustness certification for a classifier at input $\mathbf{x}$ is a neighborhood (typically {an $\ell_2$ ball}) of $\mathbf{x}$ on 
which classifier predictions are constant. Certifications aim to 
guarantee that \emph{no perturbed adversarial examples exist in this neighborhood}, including ``slightly'' perturbed instances. 

\begin{definition}\label{def:certifications}
    A \emph{pointwise robustness certification} for a $K$-class classifier $f$ at input $\mathbf{x}\in\mathbb{R}^D$
    is a real radius $R> 0$ that is sound and (optionally) complete:
    \begin{enumerate}
        \item[(i)] \emph{[sound]} $\forall \mathbf{x}'\in\mathbb{R}^D, \|\mathbf{x}'-\mathbf{x}\|\leq R \Rightarrow f(\mathbf{x}')=f(\mathbf{x})$. 
        \item[(ii)] \emph{[complete]} $\forall R'>R, \exists \mathbf{x}'\in\mathbb{R}^D, \|\mathbf{x}'-\mathbf{x}\|\leq R' \wedge f(\mathbf{x}')\neq f(\mathbf{x})$.
    \end{enumerate}
\end{definition}

For a given certification mechanism, we will distinguish the idealized certification radius $R$ (\ie the mapping of Definition~\ref{def:certifications} under the soundness condition)
from a candidate radius $\tilde{R}$ that an implementation of this mechanism computes. As we will see, the latter may not be necessarily sound (or complete).

We categorize certification mechanisms in the literature depending on their specific claims, summarized as follows:\\[0.5em]

\noindent \begin{tabular}{p{0.95\columnwidth}}
    \hline
    \textbf{Exact mechanisms}: claim to output sound and complete radii. \\[0.33em]
    \textbf{Conservative mechanisms}: claim to output sound radii that are not necessarily complete. \\[0.33em]
    \textbf{Approximate mechanisms}: claim to output random radii that are sound (or abstain), with high probability $1-\alpha$, and that are not necessarily complete. \\
    \hline
\end{tabular}\\[0.5em]

We next go into more detail on each category of certification mechanism.

\paragraph{Exact certification mechanisms.} 
These mechanisms output sound and complete radii under ideal realization of $\mathbb{R}$ arithmetic.  
Binary linear classifiers $f(\mathbf{x})=\sign(\mathbf{w}^T\mathbf{x}+b)$ admit an exact certified radius 
$R=|\mathbf{w}^T\mathbf{x}+b|/\|\mathbf{w}\|$. Cohen \etal~\cite{cohen2019certified} derive this radius and prove its 
soundness \cite[Proposition~4]{cohen2019certified}
and completeness \cite[Proposition~5]{cohen2019certified} for real arithmetic.

\paragraph{Conservative certification mechanisms.}
These are mechanisms that output radii that are sound and not necessarily complete under real-valued arithmetic.
These mechanisms usually view neural network certification as an optimization problem: for input $\mathbf{x}$ with $f(\mathbf{x})>0$, and candidate radius $R$, let
\begin{eqnarray*}
f^{\star} &=& \min_{\mathbf{x}'\in \mathbb{R}^D} f(\mathbf{x}') \\
\text{s.t.} & & \|\mathbf{x}'-\mathbf{x}\|\leq R\enspace.
\end{eqnarray*}
If $f^{\star} > 0$, then $R$ can be certified for $f$ on $\mathbf{x}$.
In general, this is a non-convex problem due to the non-linear activations used in typical choices of $f$.
Hence, conservative certification mechanisms usually relax the non-convexity of neural networks to estimate a tractable
lower bound $\underline{f} \le f^{\star}$, \eg approximating non-linear activations with linear functions~\cite{crown,beta-crown,ehlers2017formal,bunel2020branch,tjeng2017evaluating}. 
If $\underline{f} > 0$, then $f^{\star} > 0$ and $f$ is certified.
Bound propagation~\cite{crown,beta-crown} aims to estimate $\underline{f}$ to a certain input, 
by maintaining at each layer an outer approximation of the set of activation functions 
and propagating bounds from the output layer to the input layer.
Linear programming~\cite{ehlers2017formal,bunel2020branch,tjeng2017evaluating} is another approach that also relaxes non-convex constraints into linear constraints,
and solves the certification problem as a linear programming problem.
We used $\beta$-Crown and Gurobi MIP solver as example applications of bound propagation and linear programming respectively in Section~\ref{sec:crown}.

\paragraph{Approximate mechanisms.} 
Approximate certifications output random radii that under $\mathbb{R}$ are sound (or abstain), 
with high probability $1-\alpha$. Such mechanisms are not necessarily complete.
Researchers have studied approximate certification of non-linear models such as neural networks, 
under the $\ell_2$ norm~\cite{augmentation,li2019certified,cohen2019certified}.
In this work, we consider the randomized smoothing approach of Cohen \etal~\cite{cohen2019certified}.
Sufficiently stable model predictions lead to certifiable radii;
therefore we seek to stabilize the outputs of the base classifier $f$ by forming a smoothed classifier $g$ as follows.
To input $\mathbf{x}\in\mathbb{R}^D$ add isotropic Gaussian noise, then apply $f$. The input distribution induces a distribution over predictions. Smoothed classifier $g$ then
outputs the most likely class under this induced distribution. That is, for 
\begin{eqnarray*}
\boldsymbol{\epsilon} &\sim& \mathcal{N}(0, \sigma_{P}^2 I^2)\enspace, \\
g(\mathbf{x}) &=& \arg\max_{k\in [K]} \Pr\left(f(\mathbf{x}+\boldsymbol{\epsilon}) = k\right)\enspace.
\end{eqnarray*}
Consider lower (upper) bounds on the winning $k^\star$ classification's probability score, 
$$\Pr(f(\mathbf{x}+\boldsymbol{\epsilon}) = k^\star)\geq \underline{p_A} \geq \overline{p_B} \geq \max_{k\in[K]\backslash\{k^\star\}}\Pr(f(\mathbf{x}+\boldsymbol{\epsilon}=k)).$$ 
Then Cohen \etal~\cite{cohen2019certified} prove that radius $R=0.5\sigma_P (\Phi^{-1}(\underline{p_A})-\Phi^{-1}(\overline{p_B}))$ 
is a sound certification for smoothed classification $g(\mathbf{x})$, under real-valued arithmetic. 
However the floating-point calculation of a corresponding $\tilde{R}$ is likely to experience some degree of rounding, 
potentially sufficient to erroneously certify some $\tilde{R}>R$.

\section{Rounding search attack}
\label{sec:attack}
We now present a rounding search method that exploits floating-point rounding errors to
find adversarial examples within a computed certified radius $\tilde{R}$.

\paragraph{Threat model.}
Like prior works on adversarial examples~\cite{carlini2017towards, madry2017towards, jia2021exploiting}, 
we assume that the adversary has white-box access to a classifier $f$,
and has white-box access to the certification mechanism that 
it can query with inputs $f$ and instance $\mathbf{x}\in\mathbb{R}^D$, 
and obtain a certified radius~$\tilde{R}$ as an output.

Since there are floating-point rounding errors in the operations for computing a certification,
the computed radius $\tilde{R}$ at an instance $\mathbf{x}$ could overestimate an intended sound (and possibly complete) radius $R<\tilde{R}$.
This creates leeway for an adversary to find adversarial
perturbations whose norms are less than or equal to the computed certified radius, but which can
change the classifications of the model, \emph{invalidating soundness of the computed certification}.
Our work aims to find a systematic and efficient way to exploit these rounding errors.

A perturbation's norm $\|\boldsymbol{\delta}\| = \|\mathbf{x}'-\mathbf{x}\|$ 
must be estimated when evaluating the perturbation's success.
This norm computation can also suffer floating-point rounding errors, and could be underestimated.
To handle this possibility, we conduct attacks in two threat models, one weak one strong.

The weak model makes minimal assumptions on how certification is violated:
an attack is ruled successful if the floating-point computation of $\|\boldsymbol{\delta}\|$
is smaller than or equal to $\tilde{R}$ (\ie $\|\boldsymbol{\delta}\| \le \tilde{R}$).
We note that this model represents settings where the norm is computed using
common software libraries for computing operations on floating numbers 
(\eg Numpy's 32-bit or 64-bit floating-point arithmetic).

The strong model does not make these assumptions, the true (real-valued) norm of attacks must violate certifications.
This model considers a setting where the norm is computed using software packages that instead of
returning a result that is potentially rounded can return a representation of a real-valued norm or its range.
Since we cannot do real arithmetic on machines, we use the upper bound of the norm $\overline{\|\boldsymbol{\delta}\|}$ instead,
which is computed with rounded interval arithmetic and is guaranteed to be greater than or equal to the true norm.
That is, a successful attack satisfies $\overline{\|\boldsymbol{\delta}\|} \le \tilde{R}$, 
which guarantees that $\tilde{R}$ must necessarily exceed the real-valued norm of the attack.

\begin{figure}[h]
\includegraphics[width=0.4\textwidth]{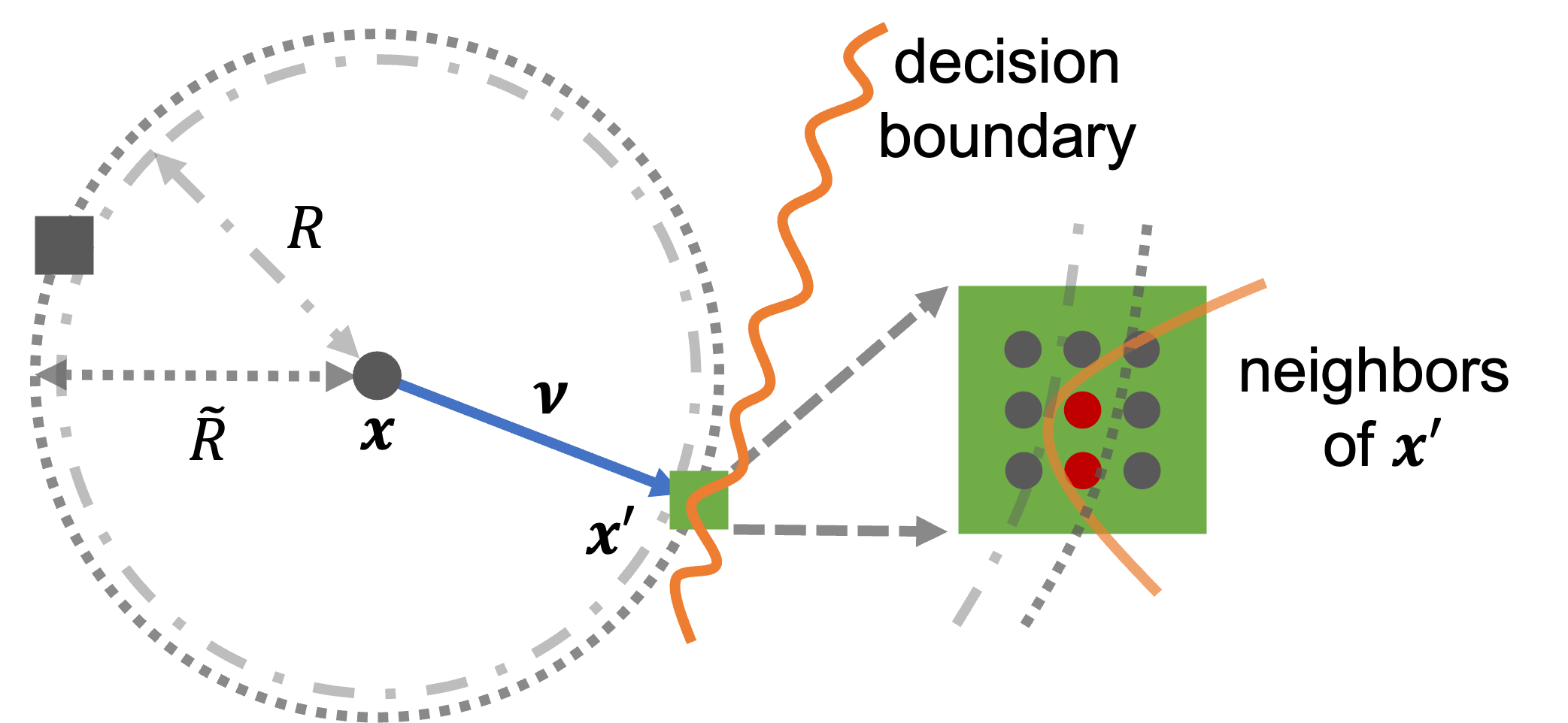}
\centering
\caption{
  The search direction $\boldsymbol{\nu}$ (blue line) 
  and search area (green area) for finding adversarial examples against a model, 
  whose decision boundary is the orange line.
  $\mathbf{x}$ is the original instance,
  $\tilde{R}$ and $R$ are the computed and real-valued certified radii of the model on $\mathbf{x}$,
  $\boldsymbol{\delta}=\tilde{R}\boldsymbol{\nu} / \|\boldsymbol{\nu}\|$ 
  is the adversarial perturbation in the search direction $\boldsymbol{\nu}$, instance 
  $\mathbf{x'}=\mathbf{x}+\boldsymbol{\delta}$ is the seed for the green search area.
  Our rounding search method will sample $N$ floating-point neighbors $\boldsymbol{\delta'}$ of $\boldsymbol{\delta}$,
  and evaluate each $\mathbf{x}+\boldsymbol{\delta'}$ to check if any one of them can 
  flip the classification of the model with $\|\boldsymbol{\delta'}\| \le \tilde{R}$ 
  or $\overline{\|\boldsymbol{\delta}'\|}\le\tilde{R}$ (the red points in the green search area).
} 
\label{fig:direction}
\end{figure}

\paragraph{Attack overview.}
Consider a classifier~$f$, input~$\mathbf{x}$ and the corresponding computed radius~$\tilde{R}$.
A na\"ive way to search for an adversarial example would be
to try all~$\mathbf{x}'$ such that $\|\mathbf{x} - \mathbf{x}'\| \le \tilde{R}$, 
checking whether~$f(\mathbf{x}) \neq f(\mathbf{x}')$. Unfortunately this exhaustive search
is computationally intractable (\eg there are $\approx 2^{17}$ floating points in a small interval such as $[10,10+2^{-32}]$). 
We can avoid {some} futile search. For example, observe that instances in the gray area, 
as depicted in Figure~\ref{fig:direction}, {are unlikely to} flip predictions, as they are in the opposite direction of the decision boundary.
A key idea is to find a perturbation direction~$\boldsymbol{\nu}$ that reaches the decision boundary in the shortest distance, 
and add a perturbation $\boldsymbol{\delta}$ in that direction to~$\mathbf{x}$,
to maximize our chance to flip the classifier's prediction 
with perturbation norm $\|\boldsymbol{\delta}\|$ (or $\overline{\|\boldsymbol{\delta}\|}$) less than or equal to $\tilde{R}$.
This baseline method has several challenges.
First, computation of perturbation direction $\boldsymbol{\nu}$ is not easy for NNs which do not typically have linear decision boundaries. 
To this end, for ReLU networks, we find a local linear approximation prior to computing the gradient for $\boldsymbol{\nu}$. 
In general one may use a standard approach to finding adversarial example \emph{directions} such as PGD.
Second, while $\boldsymbol{\nu}$ guides a search towards the decision boundary, the search may still be unable to exploit the leeway between  
the real certified radius $R$ and the computed certified radius $\tilde{R}$ to find certification violations. 
We address this challenge with a tightly-confined randomized floating-point neighborhood search.

In summary, our attack proceeds as follows (depicted in Figure~\ref{fig:direction}).

\begin{enumerate}
\item Find an adversarial perturbation direction~$\boldsymbol{\nu}$ that reaches the decision boundary of classifier~$f$
in the shortest distance, as a form of PGD attack~\cite{madry2017towards} (Section~\ref{sec:direction}).  
\item Compute perturbation $\boldsymbol{\delta}$ in the direction $\boldsymbol{\nu}$ within the computed certified radius $\tilde{R}$:
\begin{equation}
  \label{eq:scale}
  \boldsymbol{\delta}=\tilde{R}\boldsymbol{\nu} / \|\boldsymbol{\nu}\|\enspace.
\end{equation}
\item Search for multiple floating-point neighbors $\boldsymbol{\delta}'$ of $\boldsymbol{\delta}$ 
with $\|\boldsymbol{\delta}'\|\le\tilde{R}$ (or $\overline{\|\boldsymbol{\delta}'\|}\le\tilde{R}$),
and evaluate if any $\mathbf{x}+\boldsymbol{\delta}'$ can flip the classifier's prediction (Section~\ref{sec:fpsearch}).
\end{enumerate}

\subsection{Adversarial perturbation direction}
\label{sec:direction}

For linear models, direction $\boldsymbol{\nu}$ is a normal to 
the decision boundary's hyperplane $\mathbf{w}^T\mathbf{x}+b=0$ and equals $\mathbf{w}$.
The perturbation direction for neural networks is not as obvious as it is for linear models, 
as the decision boundary can be highly non-linear.
In general one can adopt a PGD-like approach for finding the attack \emph{direction}.
In the rest of this section we describe this approach for finding $\boldsymbol{\nu}$ 
for neural networks with ReLU activations that we show to be effective in our experiments. 
This is an \emph{exact} instantiation of PGD. The novelty of our attack comes later when we set the attack step size and incorporate a local floating-point search.
A neural network with ReLUs can be represented as 
$$F(\mathbf{x}) = \left(F_n \circ F_{n-1} \circ \cdots \circ F_1\right) (\mathbf{x})$$
where $F_i(\mathbf{x})=\ReLU(\boldsymbol{\theta}_i^T \mathbf{x}+\hat{\boldsymbol{\theta}}_i)$. 
Here $\mathbf{x}$ and $\hat{\boldsymbol{\theta}}_i$ are vectors, $\boldsymbol{\theta}_i$ is a matrix, 
and the rectified linear (ReLU) {activation} function acts pointwise {on a vector, returning a vector}.

We use the fact that such networks are piecewise linear: 
therefore a (local) linear approximation at instance $\mathbf{x}$ is in fact \textit{exact}.
Then, one can find an adversarial example for $\mathbf{x}$ against this linear model as described above
and use it to attack the original ReLU network.

\paragraph{Warmup.}
{As a warmup, let us consider a network where ReLUs are all activated.}
{For each node} $\ReLU(z)=\max\{0,z\}=z$, {and so the} network is a combination of $K$ linear models where
$K$ is the number of classes. That is, 
\[F(\mathbf{x}) = \boldsymbol{\theta}^T \mathbf{x}+\hat{\boldsymbol{\theta}}\enspace,\]
where $\boldsymbol{\theta}^T = \boldsymbol{\theta}_n^T \boldsymbol{\theta}_{n-1}^T \cdots \boldsymbol{\theta}_1^T$, 
and $\hat{\boldsymbol{\theta}} = \sum_{i=1}^n \left(\prod_{j=i+1}^n \boldsymbol{\theta}_j^T\right)\hat{\boldsymbol{\theta}}_i$.
Note that $\boldsymbol{\theta}^T$ is a $K\times D$ matrix and~$\hat{\boldsymbol{\theta}}$ is a column vector of length~$K$.
Each class $k$ corresponds to the linear model
\[F^k(\mathbf{x}) = \mathbf{w}_k^T \mathbf{x}+b_k\enspace,\]
where $\mathbf{w}_k^T=(\boldsymbol{\theta}^T)_{k,\cdot}$ is the $k$th row of $\boldsymbol{\theta}^T$ and ${b}_k = \hat{\theta}_k$.

In order to change this model's classification from the original class $l$ 
to the target class $t\neq l$, we observe that one can attack the following model:
\[ L(\mathbf{x}) = F^t(\mathbf{x}) - F^l(\mathbf{x}) = (\mathbf{w}_t^T - \mathbf{w}_l^T) \mathbf{x} + {b}_t - {b}_l\enspace.\]
This is a linear model, and $L(\mathbf{x}) < 0$ when $F(\mathbf{x})$ classifies $\mathbf{x}$ as~$l$,
$L(\mathbf{x}) > 0$ when $F(\mathbf{x})$ classifies $\mathbf{x}$ as $t$, so $L(\mathbf{x})$ has the decision boundary hyperplane $L(\mathbf{x}) = 0$.
Hence, the most effective perturbation direction to change classification of $F(\mathbf{x})$ from $l$ to $t$, 
as before for linear models, is~$\boldsymbol{\nu} = \mathbf{w}_t^T - \mathbf{w}_l^T$, 
which is the gradient of $L(\mathbf{x})$ with respect to $\mathbf{x}$.

\paragraph{Linear approximation of ReLU networks.}

ReLUs will all be activated when the weights and biases of each hidden layer 
of the network are positive, and all values of the input are also positive 
(\eg an image, whose pixel value is usually in the range $[0,1]$).
However, in practice this usually is not the case and some ReLUs will not be activated.
For inactive ReLUs, we modify outgoing weights to zero in the calculation of the perturbation direction $\boldsymbol{\nu}$.

The overall process,~$\Perturb$, is described in~Algorithm~\ref{alg:perturb}. 
It proceeds by first finding an exact (local) linear approximation 
$F'(\mathbf{x}) = \boldsymbol{\tau}^T \mathbf{x} + \hat{\boldsymbol{\tau}}$ 
where $\hat{\boldsymbol{\tau}} \gets \sum_{i=1}^n \left(\prod_{j=i+1}^n \boldsymbol{\tau}_j^T\right)\hat{\boldsymbol{\theta}}_i$ using the notation in the pseudo-code.
The weights of~$F'$ are equal to weights of $F$ for internal nodes where~$F(\mathbf{x})$ activated the corresponding ReLUs,
otherwise they are set to 0. Specifically, we zero out columns of matrix $\boldsymbol{\theta}_i^T$ 
when the corresponding elements of mask $\mathbf{m}_i$ are zero. Given these weights, $\Perturb$ computes~$\boldsymbol{\nu}$ as explained in the warmup.
This direction corresponds to a gradient of the network's target minus current class scores, with respect to the instance $\mathbf{x}$.

\DontPrintSemicolon
\begin{algorithm}[h]
  \caption{$\Perturb$: Linearized ReLU Network Gradients}\label{alg:perturb}
  \KwIn{
  input to be perturbed $\mathbf{x}$;
  the neural network model $F(\mathbf{x}) =  F_n \circ F_{n-1} \circ \cdots \circ F_1(\mathbf{x})$,
  $F_i(\mathbf{x})=\ReLU(\boldsymbol{\theta}_i^T \mathbf{x}+\hat{\boldsymbol{\theta}}_i)$;
  current label $l$; adversarial target label $t$.
  }
  \KwOut{$\boldsymbol{\nu}$, a perturbation direction.}

  \SetKwFunction{FMain}{$\Perturb$}
      \SetKwProg{Fn}{Function}{:}{}
      \Fn{\FMain{$\mathbf{x}$, $F$, $l$, $t$}}{
        $\mathbf{h}_1 \gets \boldsymbol{\theta}_1^T \mathbf{x}+\hat{\boldsymbol{\theta}}_1$\;
        \For{$i\leftarrow 2$ \KwTo $n$}{
          $\mathbf{m}_i \gets  \mathbbm{1}_{[\mathbf{h}_{i-1} > 0]}$ \quad\Comment{{elementwise thresholding}}
          $\boldsymbol{\tau}_i^T \gets \boldsymbol{\theta}_i^T \odot \mathbf{1}\mathbf{m}_i^T$ \quad\Comment{{Hadamard product; $\mathbf{1}$ is a column vector of 1s}}
          $\mathbf{z}_{i-1} \gets \ReLU(\mathbf{h}_{i-1})$\; 
          $\mathbf{h}_i \gets \boldsymbol{\theta}_i^T \mathbf{z}_{i-1}+\hat{\boldsymbol{\theta}}_i$\;
        }
        $\boldsymbol{\tau}^T \gets \boldsymbol{\tau}_n^T \boldsymbol{\tau}_{n-1}^T \cdots \boldsymbol{\tau}_2^T \boldsymbol{\theta}_1^T$\qquad\Comment{weights of $F'$}
      
        $\boldsymbol{\nu} \gets \boldsymbol{\tau}^T[t] - \boldsymbol{\tau}^T[l]$\;
        
        \textbf{return} $\boldsymbol{\nu}$
  }
  \textbf{End Function}
\end{algorithm}

\paragraph{Projected gradient descent for ReLU networks.}

Given $\boldsymbol{\nu}$ as output by Algorithm~\ref{alg:perturb} and a computed certified radius, 
one could compute adversarial perturbation $\boldsymbol{\delta}$ in direction $\boldsymbol{\nu}$ close to the certified radius as in Equation~\ref{eq:scale}.
However, the resulting $\mathbf{x}'=\mathbf{x}+\boldsymbol{\delta}$ may activate different ReLUs of~$F$ than~$\mathbf{x}$. 
Hence, the linear approximation $F'$ on~$\mathbf{x}'$ may be different to~$F'$ on~$\mathbf{x}$: these approximations are only exact in local neighborhoods.
To this end we perform a search by iteratively updating~$\mathbf{x}'$ and 
invoking $\Perturb$ until an adversarial example within the input domain $[V_\mina, V_\maxa]$ is found or the procedure times out.
Algorithm~\ref{alg:weight} describes this procedure, which we refer to as $\PDAW$.
The algorithm iteratively performs the following: computes the gradient of the network's linearization at the current iteration, rescales to the step size $s$, 
clips the perturbation to the domain constraint, applies the perturbation.

\SetKwRepeat{Do}{do}{while}
\begin{algorithm}[h]
  \caption{$\PDAW$: Linearized Projected Gradient Descent for ReLU NNs}\label{alg:weight}
  \KwIn{
  input to be perturbed $\mathbf{x}$;
  the neural network model $F(\mathbf{x}) =  F_n \circ F_{n-1} \circ \cdots \circ F_1$,
  $F_i(\mathbf{x})=\ReLU(\boldsymbol{\theta}_i^T \mathbf{x}+\hat{\boldsymbol{\theta}}_i)$;
  current label $l$; adversarial target label $t$;
  step size $s$;
  input domain $[V_{\mina}, V_{\maxa}]$.
  }
  \KwOut{$\boldsymbol{\delta}$, adversarial perturbation.}

  \SetKwFunction{FMain}{$\PDAW$}
      \SetKwProg{Fn}{Function}{:}{}
      \Fn{\FMain{$\mathbf{x}, F, l, t, \step, V_{\mina}, V_{\maxa}$}}{
       $\mathbf{x}' \gets \mathbf{x}$ \quad\Comment{initial adversarial example}
        \Do{$F(\mathbf{x}') \neq t$\quad\mbox{\Comment{or till timeout}}}{
          $\boldsymbol{\nu} \gets \Perturb(\mathbf{x}', F, l, t)$\;
          $\boldsymbol{\delta}\gets\frac{\step}{\|\boldsymbol{\nu\|}}\boldsymbol{\nu}$\;
          $\mathbf{x}' \gets \mathbf{x}' + \boldsymbol{\delta}$\;
          $\mathbf{x}' \gets \clip(\mathbf{x}', V_{\mina}, V_{\maxa})$
        }
        \textbf{return} {$\boldsymbol{\delta} \gets \mathbf{x}' - \mathbf{x}$}
      }
  \textbf{End Function}
\end{algorithm}

\begin{remark}
Note that $\tilde{R}$ may not be given, as is the case for some network verifiers that instead of returning~$\tilde{R}$,
take $F$, $\mathbf{x}$ and some $R$ as input and either certify $R$ or not.
In this case, we need to search for the smallest perturbation in the direction of $\boldsymbol{\nu}$ to find such an $R$ to attack.
Hence, in Algorithm~\ref{alg:weight} we use $\step$ as an input,
which is set to a small initial value (\eg $10^{-5}$ in our experiments) so that $\boldsymbol{\nu}$ can be updated frequently.
If a $\tilde{R}$ is given, we can set it as a threshold value to stop the algorithm, that is,
the algorithm should stop when the total perturbation norm reaches $\tilde{R}$.
\end{remark}

\subsection{Rounding search}
\label{sec:fpsearch}

Given the direction $\boldsymbol{\nu}$ and the computed certified radius $\tilde{R}$, 
an adversarial perturbation $\boldsymbol{\delta}$ can be computed using Equation~\ref{eq:scale},
and $\mathbf{x}'=\mathbf{x} + \boldsymbol{\delta}$ should give an adversarial example so that $F(\mathbf{x}) \neq F(\mathbf{x}')$.

If the accumulated rounding errors are large, $\boldsymbol{\delta}$ can be sufficient to conduct a successful attack (\eg for neural networks with many neurons).
For some attacks, the rounding errors we exploit are much smaller, such as linear models with fewer operations.
Hence, we create $N$ floating-point neighbors of~$\boldsymbol{\delta}$ 
{to explore more possibilities of robustness violations close to the decision boundary due to rounding errors}.
At a high level, each neighbor $\boldsymbol{\delta'}$ is constructed by using $\boldsymbol{\delta}$ as a seed and then, for each dimension, replacing the original value
with a {neighboring} floating point that is either larger or smaller than it.
For example, a neighbor of $[1.0, 1.0]$ can be $[0.9999999999999999, 1.0000000000000002]$.
We provide the pseudo-code of {the neighbors sampling} procedure in~Algorithm~\ref{alg:neighbor}. 
We call this algorithm $\Neighbor$.
The result is a set of $N$ neighboring perturbations.
Then for each neighbor $\boldsymbol{\delta}'$ we test if $\mathbf{x}+ \boldsymbol{\delta}'$ 
leads to an adversarial example {(\ie flips the classifier's prediction)} that is certified 
(\ie $\|\boldsymbol{\delta}'\| \le \tilde{R}$ in the weak threat model, 
or $\overline{\|\boldsymbol{\delta}'\|} \le \tilde{R}$ in the strong threat model).

\newcommand{\neighbn}{{p}}
\begin{algorithm}[h]
  \caption{$\Neighbor$: FP Neighbors Search}\label{alg:neighbor}
  \KwIn{
  perturbation seed vector $\boldsymbol{\delta}=[\delta_1, \delta_2, \ldots, \delta_D]$;
  \\ $\neighbn$: number of {signed} neighbor values to sample from for each dimension;
  \\  $N$: number of {sampled} neighbors of $\boldsymbol{\delta}$ to return;
  }
  \KwOut{$\boldsymbol{\delta}$-$\neighb$, $N$ neighbors of $\boldsymbol{\delta$}}%
  \SetKwFunction{FMain}{$\Neighbor$}
  \SetKwProg{Fn}{Function}{:}{}
  \Fn{\FMain{$\boldsymbol{\delta}$, $N$, $n$}}{
    $\candidates \gets [{(2\neighbn+1)} \times D]$\;
    \For{$i\leftarrow 1$ \KwTo $D$}{
      $\candidates[1][i]\gets \delta_i$\;
      $r_i \gets \delta_i$\;
      $l_i \gets \delta_i$\;
      $j\gets 2$\;
      \While{$j < {2\neighbn+1}$}{
          $r_i \gets \rightp(r_i)$\;
          $l_i \gets \leftp(l_i)$\;
          $\candidates[j][i]\gets r_i$\;
          $\candidates[j+1][i]\gets l_i$\;
          $j\gets j+2$\;
      }
    }
    $\boldsymbol{\delta}$-$\neighb \gets \sample(\candidates, N)$ ~\Comment{Construct $N$ neighbor points: for each point $\boldsymbol{\delta'}$, randomly sample component $\delta_i'$ from $\candidates[\cdot][i]$, such that $\boldsymbol{\delta}$-$\neighb$ contains no duplicates nor copies of $\boldsymbol{\delta}$.}
    }
  \textbf{return} $\boldsymbol{\delta}$-$\neighb$\;
  \textbf{End Function} 
  \end{algorithm}

In Algorithm~\ref{alg:neighbor}, $\leftp(v)$ returns the first floating-point value before $v$,
and $\rightp(v)$ returns the first floating-point value after $v$.  
In the experiments, we set $\neighbn=2$ for Algorithm~\ref{alg:neighbor}.
\section{Attack experiments}
\label{sec:att:exp}

In this section we evaluate whether our rounding search attacks 
can find adversarial examples within a certified radius.
We first consider linear classifiers and then neural networks. 
We evaluate certified radii obtained using the \emph{exact} method for linear classifiers, 
and \emph{conservative}~\cite{beta-crown,gurobi}
and \emph{approximate}~\cite{cohen2019certified} certification mechanisms for neural networks.
Since the computation of exact certification for linear classifier is $R=|\mathbf{w}^T\mathbf{x}+b|/\|\mathbf{w}\|$, 
we compute it ourselves using either 32-bit or 64-bit floating-point arithmetic in Numpy.

For linear classifiers, we will conduct attacks in both the weak and strong threat models.
For neural networks, we will conduct attacks in the weak threat model.
We show that our rounding search finds adversarial examples within certified radii for all of them.
Our linear models are run on an Intel Xeon Platinum 8180M CPU, 
while our neural network models are run on a Tesla V100 16GB GPU.

\paragraph{Baseline attack rates.}
The baseline success rate for finding an adversarial example against a linear model within 
the radius defined in~Section~\ref{sec:bqcert} should be 0\% in both threat models,
since the mechanism is exact: it claims to be both sound and complete. 
The baseline success rate for radii returned by conservative mechanisms should also be 0\%
since they too are claimed to be sound. 
Though randomized smoothing comes with a failure probability $\alpha\ll 1$ to account
for sampling error in approximating a smoothed classifier,
it does not explicitly take into account errors due to rounding.

\paragraph{Model training.}
We train (primal) linear SVM with sequential minimal optimization, by using corresponding modules of scikit-learn.
Our linear classifiers are trained with $\ell_1$ regularization 
so that model weights are sparse, and perturbations are less likely to move
images outside their legal domain
(recall that the perturbation direction for a linear classifier is its weights $\boldsymbol{\nu} = \mathbf{w}$).
All ReLU networks in this section are trained with the SGD optimizer using PyTorch, 
with momentum 0.9, learning rate 0.01, batch size 64, for 15 epochs.
For some controlled experiments we require the weights and biases of the hidden layers to be positive (to activate all ReLUs).
In this case weights and biases are clamped with the lower bound 0 
after each step of training.

\begin{figure*}[t]
\centering
\begin{subfigure}{0.5\textwidth}
\centering
    \includegraphics[width=0.8\linewidth]{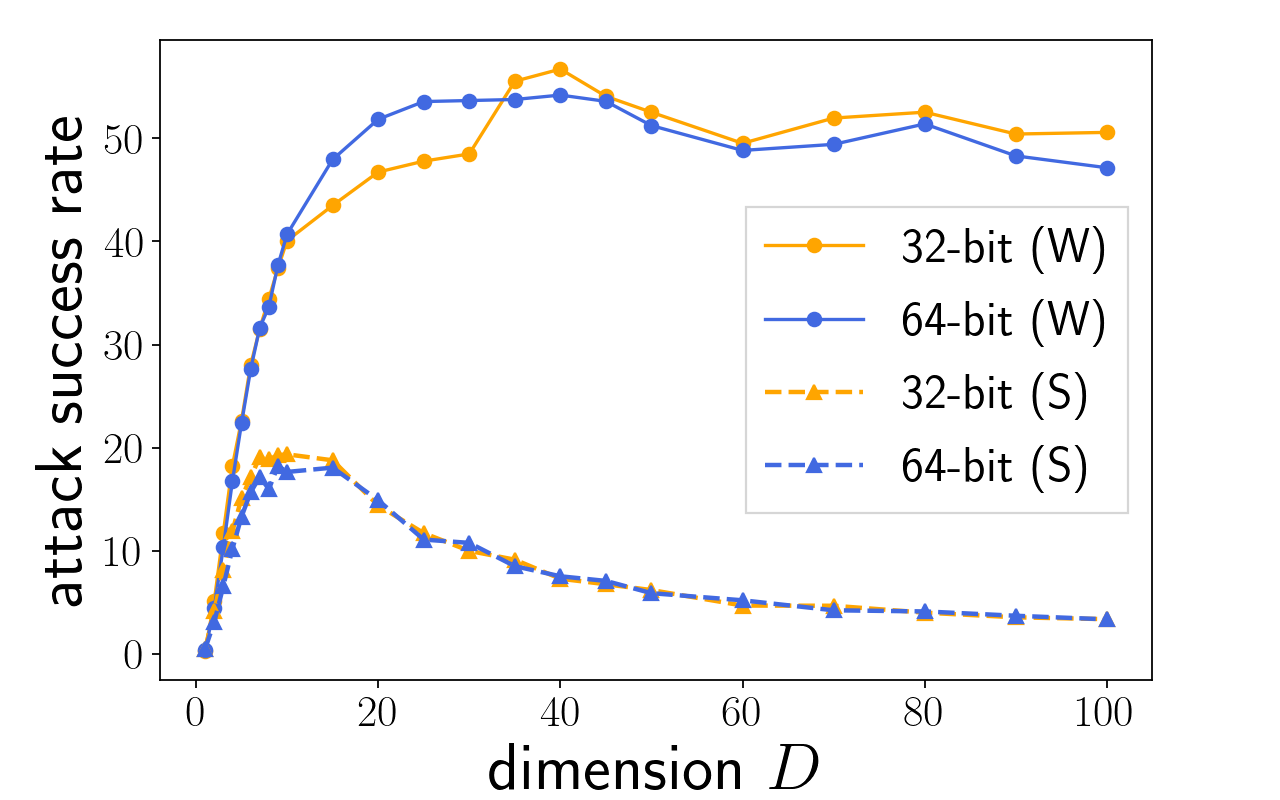}
    \caption{}
    \label{fig:linear_succ}
\end{subfigure}%
\begin{subfigure}{0.5\textwidth}
\centering
    \includegraphics[width=0.8\linewidth]{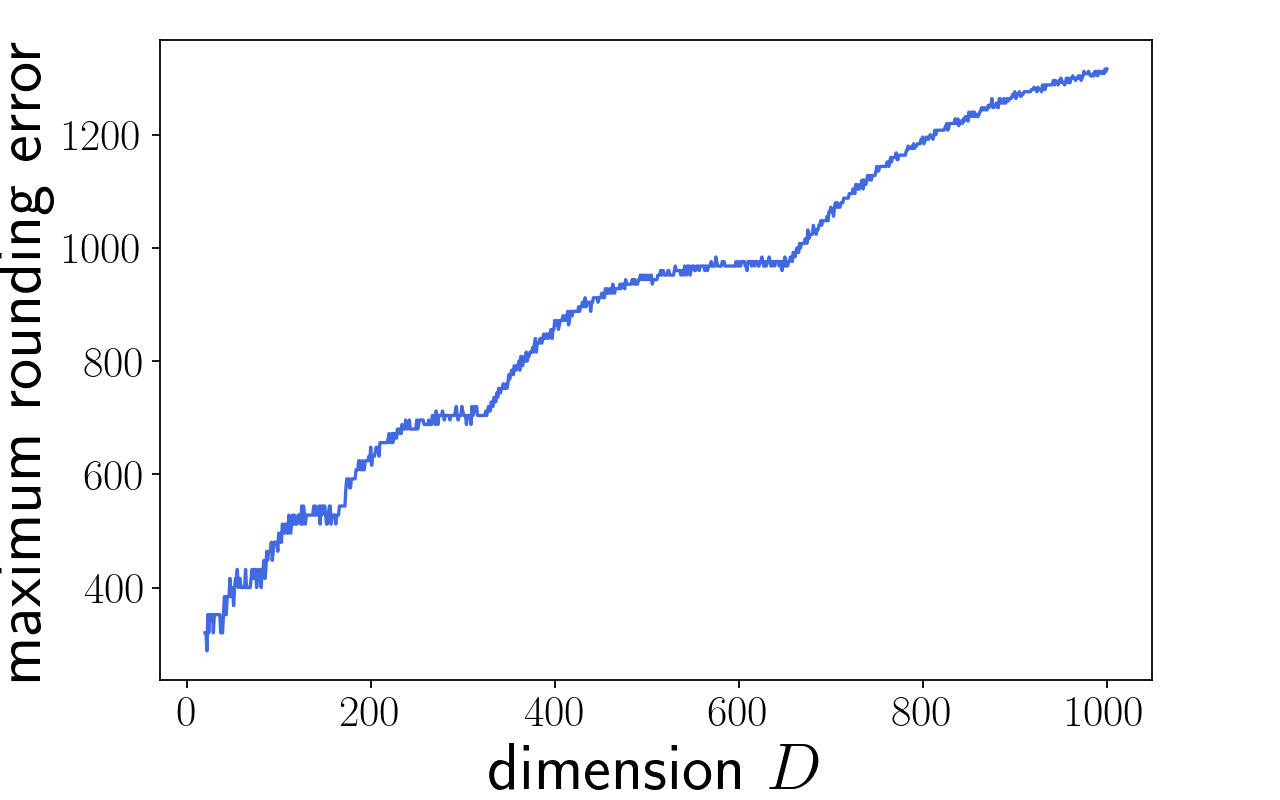}
    \caption{}
    \label{fig:r_error}
\end{subfigure}
\caption{
(a)
Rounding search attack success rates against a random binary linear classifier 
in both weak (W) and strong (S) threat models (Section~\ref{sec:exp:linear}).
For each dimension $D$, we report the {percentage} of 10 000 randomly initialized models for which we can successfully
find an adversarial example within certified radius $\tilde{R}$ for a random instance~$\mathbf{x}$ drawn from $[-1,1]^D$.
Since the attacks are against an exact certified radius,
the baseline attack rate should be~$0\%$ in both weak and strong threat models.
Model weights $\mathbf{w}$ and biases $b$
are randomly initialized with $\mathbf{w}\in[-1,1]^D$, $b\in[-1,1]$.
All values and computation is done using either 32-bit or 64-bit floating points.
(b)
Maximum rounding error in the calculation of 
the certified radius $\tilde{R}$ on a sample $\mathbf{x}$ with each $x_i=3.3\times10^{9}$, for the linear model with
$w_i=3.3\times10^{-9}$, $b=3.3\times10^9$, where $i \in [1,D]$ and $D\in[20,1000]$.
}
\label{fig:random_att}
\end{figure*}

\subsection{{Random linear classifiers}}
\label{sec:exp:linear}

To evaluate the performance of our attack in an ideal scenario,
we first conduct our attack on randomly initialized {(binary) linear classifiers with randomly generated target instances}: $f(\mathbf{x})=\sign(\mathbf{w}^T\mathbf{x}+b)$,
where weights {${w}_i$} and bias $b$ are random values drawn from the range $[-1,1]$, $\forall i \in [D]=\{1,\ldots,100\}$.
Each value is represented with either 32-bit or 64-bit floating-point precision.
For each dimension, we test 10 000 randomly initialized models. For each model
we choose one {instance}~$\mathbf{x}$ to attack, where each {component} {${x}_i$} is drawn randomly from $[-1,1]$. Hence, 
attack success rate measures the number of models out of 10 000 for which a random instance can result in a successful attack.
For each combination of $(\mathbf{w},b,\mathbf{x})$, we sample and evaluate {$N=D^2$} {neighboring perturbations} {of $\boldsymbol{\delta}=\tilde{R}\mathbf{w}/\|\mathbf{w}\|$}
using the $\Neighbor$ function (Algorithm~\ref{alg:neighbor}).

Results are shown in Figure~\ref{fig:linear_succ}.
With higher dimension, our attack success rate first increases and then flattens around~50\% in the weak threat model, and around~5\% in the strong threat model.
We investigate the flattening phenomenon in Section~\ref{app:revisit}.

With higher dimension more arithmetic operations are done in computing $\|\boldsymbol{\delta}'\|$ and {$\tilde{R}$}, 
which results in accumulation of rounding errors.
Figure~\ref{fig:r_error} further shows this influence of $D$ on the rounding error, 
which can be accumulated to the magnitude of $10^{3}$ with increasing $D$.
In summary, with the increasing rounding error,
a greater leeway is left between the real certified radius $R$ and 
the computed certified radius $\tilde{R}$ for our method to exploit,
so the attack success rate increases.

The success rates are lower in the strong threat model than in the weak threat model.
This is expected, as the leeway (\ie $\tilde{R}-\overline{\|\boldsymbol{\delta}\|}$) 
exploited by our attack in the conservative strong model is likely much smaller than
that (\ie $\tilde{R}-\|\boldsymbol{\delta}\|$) in the weak model.
Why does the success rate first increase ($D\le15$) and then steadily drop ($16\le D\le100$) in the strong threat model (the triangle-marked dashed curves of Figure~\ref{fig:linear_succ})?
The reason is that we use the upper bound of the perturbation norm $\overline{\|\boldsymbol{\delta}\|}$ estimated with interval arithmetic in the strong model.
With increasing dimension, more operations are involved, and interval arithmetic may become more uncertain about its estimation. When interval arithmetic is less certain, it outputs a larger interval so as to contain the real result.
That is, the upper bound $\overline{\|\boldsymbol{\delta}\|}$ gets larger.
Hence, although $\tilde{R}$ does get larger with increasing dimension,
the leeway $\tilde{R}-\overline{\|\boldsymbol{\delta}\|}$ exploited by our attack in the strong model
first gets larger and then gets smaller, and our success rate decreases after first increasing.

\subsection{{Linear SVM}}
\label{sec:lrmodel}
In this section, we evaluate our attack on linear SVM trained with the MNIST dataset.
MNIST~\cite{mnist} contains images of hand-written digits where each image has~784 attributes and each attribute is a pixel {intensity} within the domain $[0,255]$.
We used $\approx$ 12 000 images for training, and $\approx$ 2 000 images
for validation and evaluation of our attacks,
for each combination of the labels $i,j\in\{0,\ldots,9\}$.
We trained 45 models for each combination of {distinct} labels $i,j\in\{0,\ldots,9\}$ of the MNIST dataset.
Validation accuracies range between 91\% and 99\% for linear SVM.

We then try to find an adversarial image {with respect to} each image in the test dataset.
Our attack samples $N=5$ 000 neighbors {of {$\boldsymbol{\delta}=\tilde{R}\mathbf{w}/\|\mathbf{w}\|$}} 
using Algorithm~\ref{alg:neighbor}.

In the weak threat model, we observe non-zero attack success rates for 44/45 models 
(full results appear in Table~\ref{tab:svm_mnist_res_weak} of Appendix~\ref{app:res}),
and our attacks can have success rates up to 23.24\%.
In the strong threat model, we observe non-zero attack success rates for 11/45 models 
(full results appear in Table~\ref{tab:svm_mnist_res_strong} of Appendix~\ref{app:res}),
and our attacks can have success rates up to 0.16\%.
Recall that the baseline success rate should always be 0\%.
We demonstrate a weak model example of original and adversarial images together with their perturbation norm
and certified radius information in Figure~\ref{fig:lr_mnist_demo} of Appendix~\ref{app:res}.

\subsection{Certification for {neural nets}}
\label{sec:crown}

We now turn our attention to neural network verification mechanisms.
In this section we consider neural networks with ReLU {activations} and rely on {their} linear approximations.
Given a radius $\tilde{R}$, a neural network $F$ and an input $\mathbf{x}$, these mechanisms 
either certify $\tilde{R}$ or not.
Hence, in order to find a tight certified radius for a given model,
one can perform a binary search to check multiple radii and call those verifiers multiple times.
We avoid the binary search to find a certified radius $\tilde{R}$
by first finding an adversarial example $\mathbf{x}'$ via $\PDAW$ (Section~\ref{sec:direction} and Algorithm~\ref{alg:weight})
and then trying to verify the perturbation norms ({\ie} $\|\mathbf{x}'-\mathbf{x}\|$) of those adversarial examples using the verifiers.
We set $\PDAW$ to time out after 15 minutes.

\paragraph{Certification with $\beta$-CROWN}
\label{sec:betann}
{$\beta$-CROWN}~\cite{beta-crown} guarantees sound but not complete robustness certification.
That is, it provides a lower bound on the radius and it is possible that a tighter (larger) radius may exist.
We use {the} $\beta$-CROWN verifier~\cite{beta-crown} in the $\ell_2$ metric, 
to verify a 3-layer neural network binary classifier {with 1 node in the hidden layer}.
All model weights and biases in the hidden layers of this classifier are trained to be positive,
so the perturbation direction is always $\boldsymbol{\nu}=\mathbf{w}_t^T - \mathbf{w}_l^T$.
{The classifier has validation accuracy 99.67\%.}
We use $\PDAW$, with step size $\step=1\times10^{-5}$, to incrementally add perturbation
in direction $\boldsymbol{\nu}$ to image $\mathbf{x}$, until its prediction is flipped, and we get $\mathbf{x}'$.
Then we use $\beta$-CROWN to verify the image {with respect to} $\|\boldsymbol{\delta}\| = \|\mathbf{x}'-\mathbf{x}\|$.
If the certification succeeds, we have a successful attack.
$\PDAW$ times out on {30 out of 2 108} images.
{When the attack does not time out, it takes $\approx 30$ seconds. A call to a verifier takes $\approx 1$ second.} 
We conduct our attack on all MNIST test images labeled 0 or 1.
We find adversarial images for 2 078 images, and $\beta$-CROWN erroneously verifies~53 of them.
Our attack success rate is 2.6\%.

\paragraph{Certification with MIP solver.}
\label{sec:mip-res}
We now consider another method that provides conservative verification via mixed-integer programming (MIP).
We use the implementation from~\cite{beta-crown}, 
which uses the Gurobi MIP solver~\cite{gurobi} for neural network verification, 
with Gurobi feasibility tolerance $2\times10^{-5}$.
We verify two 3-layer neural network multiclass classifiers {with 100 nodes in their hidden layer}.
For the first classifier, the weights and biases in the hidden layers are all trained to be positive,
so the perturbation direction is $\boldsymbol{\nu} = \mathbf{w}_t^T - \mathbf{w}_l^T$.
The second classifier is trained without constraints on its weights and represents a regular {network} without artefacts.
The validation accuracies for the first and second classifiers are 84.14\%, and 96.73\%, respectively.

We use $\PDAW$ to attack the two classifiers on all images of the MNIST test dataset, with step size $\step=1\times10^{-5}$.
We found adversarial images against 8 406 images for the first classifier, 
and adversarial images against 9 671 images for the second classifier.
$\PDAW$ times out on {only} 8 and 2 images for the first and second classifier, respectively.
We then use MIP to verify each image {with respect to} their adversarial image's perturbation norm~$\|\boldsymbol{\delta}\|$.
Each attack takes $\approx 30$ seconds and each verification takes $\approx 10$ seconds.
MIP successfully verified 5 108 out of 8 406 successfully attacked images for the first classifier,
and verified 1 531 out of 9 671 successfully attacked images.
That is, the attack success rate is 60.76\% on an artificially trained network (where ReLUs are {all} activated) 
and~15.83\% on the second classifier trained without artefacts.

\paragraph{Certification with randomized smoothing.}
\label{sec:rs-res}
Our experiments on approximate certification attack a 3-layer ReLU network with 100 nodes in the hidden layer 
trained on the MNIST~\cite{mnist} dataset with labels 0 and 1, which has validation accuracy 99.95\%. 
Pixel values of all images are scaled to $[0, 1]$.

We adopt the same prediction procedures as~\cite{li2019certified}.
Given an instance $\mathbf{x}$, the smoothed classifier $g$ runs the base classifier $f$ 
on $M$ noise-corrupted instances of $\mathbf{x}$, and returns the top class $k_A$ that has been predicted by $f$.
The estimation of the certified radius $\tilde{R}$ for a smoothed classifier $g$
is based on the probability distribution of the winner and runner-up classes (\ie $p_A$ and $p_B$).
$p_A$ and $p_B$ are estimated via interval estimation~\cite{binomial_prop}, with confidence $1-\alpha$.
As recommended in these papers, we set $\alpha=0.1\%$, Gaussian prediction noise scale $\sigma_P\in\{1.0, 3.0, 5.0, 7.0\}$, 
and let the number of Monte Carlo {samples} $M$ range in $\{100, 1000, 10000\}$.

We conduct our attack as follows.
Given $\mathbf{x}$ and $\tilde{R}$ as returned by $g$, we first use $\PDAW$ (maximum runtime is set as 15 minutes) to find an adversarial perturbation $\boldsymbol{\delta}$ against the base classifier $f$,
then use $\Neighbor$ to find $N=1000$ neighbors $\boldsymbol{\delta}'$ of $\boldsymbol{\delta}$.
We evaluate all generated adversarial examples $\mathbf{x}+\boldsymbol{\delta}'$ using the robust classifier $g$.
For an image $\mathbf{x}$, if any one $\mathbf{x}+\boldsymbol{\delta}'$ of the $N=1000$ adversarial images flips the classification of the model with {$\|\boldsymbol{\delta}'\| \le \tilde{R}$},
we have a successful attack against the approximate certification on this image.
We chose $f$ in the process of perturbation generation as we need a concrete ReLU network with clear weights.
We conduct attacks on 2 115 images of the MNIST dataset with labels 0 and 1, and have success rates up to 21.11\% against approximate certification. 
Detailed results are listed in Table~\ref{tab:rs_res} of Appendix~\ref{app:res}.

\section{Mitigation: Certification with rounded interval arithmetic}
\label{sec:mitigation}

Our attack results demonstrate that floating-point rounding invalidates the soundness claims of 
a wide range of certification implementations for a variety of common models.
How might such rounding errors in certification calculations be mitigated?

Rounding errors violating certifications are sometimes small.
For example, the rounding error for the certified radius of the first MNIST image of Figure~\ref{fig:lr_mnist_demo} of Appendix~\ref{app:res} is in the 13th decimal place.
One's first intuition may be to adopt slightly more conservative radii (\eg using $\tilde{R}-\gamma$ for some positive constant $\gamma$). 
Unfortunately, such radii are not in general sound, and attacks against $\tilde{R}-\gamma$ are still possible.
For example, for the same setting as in Section~\ref{sec:exp:linear}, 
it is easy to construct a linear classifier and find adversarial examples against it within 
$\tilde{R}-1.0$. As we show in Figure~\ref{fig:r_error}, it is possible to construct instance-model pairs that undergo significant rounding errors. 
Just as a certification $R$ should be sound when assuming real arithmetic, an implemented certification $\tilde{R}$ must be sound under floating-point computation. 
Constant corrections $\gamma$ and similar mitigations that are not sound are therefore inappropriate for certifying robustness.

We outline a mitigation applying rounded interval arithmetic~\cite{higham2002accuracy} to certified robustness.
Interval arithmetic is an approach to bounding approximations in numerical analysis. 
It involves replacing each (possibly) approximate floating-point scalar with an interval 
with floating-point end-points that are guaranteed to bound the scalar. 
Our goal is to enable any computation on floating-points to be possible on interval representations, 
such that this `guaranteed bounding' property on input data, parameters, or intermediate computations, 
is invariant to further computation.

It is first necessary to extend the standard arithmetic operators to interval arithmetic. 
The following definition demonstrates this process assuming arithmetic on $\mathbb{R}$: 
provided two target values $a,b$ are contained in intervals $[a_1,a_2], [b_1,b_2]$ to begin with, 
then the presented operators are guaranteed to \emph{maintain} this property on basic arithmetic operations. 
For example, $a+b\in ([a_1,a_2]+[b_1,b_2])$, where the operator `+' on reals is overloaded to real intervals.

\begin{definition}\label{def:arithmetic-operators}
For real intervals $[a_1,a_2], [b_1,b_2]$, define the following interval operators for elementary arithmetic:
\begin{itemize}
  \item Addition $[a_1,a_2] + [b_1,b_2]$ is defined as $[a_1+b_1, a_2+b_2]$.
  \item Subtraction $[a_1,a_2] - [b_1,b_2]$ is defined as $[a_1-b_2, a_2-b_1]$.
  \item Multiplication $[a_1,a_2] * [b_1,b_2]$ is defined as $[\min\{a_1b_1,a_2b_1,\\a_1b_2,a_2b_2\}, \max\{a_1b_1,a_2b_1,a_1b_2,a_2b_2\}]$.
  \item Division $[a_1, a_2] / [b_1, b_2]$ is defined as $[a_1,b_2] \times \frac{1}{[b_1, b_2]}$ where by cases:
  \begin{eqnarray*}
    \frac{1}{[b_1, b_2]} &=& \begin{cases}
            \left[ \frac{1}{b_2}, \frac{1}{b_1}\right]\enspace, & \mbox{if } 0\notin [b_1, b_2] \\
            \left[ -\infty, \frac{1}{b_1}\right] \cup \left[ \frac{1}{b_2}, \infty\right]\enspace, & \mbox{otherwise}
    \end{cases} \\
    \frac{1}{[b_1, 0]} &=& \left[ -\infty, \frac{1}{b_1}\right], \enspace
    \frac{1}{[0, b_2]} = \left[ \frac{1}{b_2}, \infty\right]
  \end{eqnarray*}
\end{itemize}
\end{definition}

\emph{Rounded} interval arithmetic~\cite{higham2002accuracy} further employs {floating-point} rounding 
when implementing the arithmetic operators, to achieve these sound floating-point extensions.

\begin{lemma}\label{lem:sound-arithmetic}
Consider the interval arithmetic operators in Definition~\ref{def:arithmetic-operators} with
the resulting lower (upper) interval limits computed using IEEE754 floating-point arithmetic with rounding down (up),
then the resulting \emph{rounded interval arithmetic operators} are sound floating-point extensions.
\end{lemma}

That is, given a collection of floating-point intervals representing a collection of corresponding reals, 
rounded extension operators produce floating-point intervals that are guaranteed 
to contain the result of corresponding (base, unextended) arithmetic on the given collection of reals.

By representing constants as singleton intervals, rounded interval arithmetic computes 
floating-point bounds on real{-valued} results.
Applying Definition~\ref{def:arithmetic-operators} and Lemma~\ref{lem:sound-arithmetic}, we can 
compute a sound floating-point interval for the rational $1/3$:
\begin{align*}
  \frac{[1,1]}{[3,3]} & = [1,1] * \frac{1}{[3,3]} = [1,1] * \left[\frac{1}{3}, \frac{1}{3}\right] \\
  & = \left[\left\lfloor{\frac{1}{3}}\right\rfloor, \left\lceil{\frac{1}{3}}\right\rceil\right] = [0.33\dots33, 0.33\dots337] 
\end{align*}

Beyond rounded interval arithmetic operators, libraries such as PyInterval offer 
rounded interval implementations of standard algebraic functions (\eg the square root) and transcendental functions 
(\eg the exponential, logarithm, trigonometric, and hyperbolic functions) 
using Newton-Raphson approximation that itself applies the above basic operators. 
Such functionality enables application of rounded interval arithmetic to general machine learning models and their certifications.

\begin{theorem}\label{thm:mitigation}
    Consider a classifier $f$, floating-point instance $\mathbf{x}$, and a certification mechanism $R(f,\mathbf{x})$ that 
    is sound when employing real arithmetic. 
    If $R(f,\cdot)$ can be computed by a composition of real-valued operators $\psi_1,\ldots,\psi_L$ with sound 
    floating-point extensions $\phi_1,\ldots,\phi_L$, then the {following} certification mechanism 
    $\underline{R}(f,\mathbf{x})$ is sound with floating-point arithmetic under the strong threat model: 
    run the compositions of $\phi_1,\ldots,\phi_L$ on coordinate-wise intervals $[\mathbf{x},\mathbf{x}], [f(\mathbf{x}),f(\mathbf{x})]$ to
    obtain $[\underline{R},\overline{R}]$;
    return $\underline{R}$.
\end{theorem}

\begin{proof}
The result follows by strong induction on the levels of composition 
implementing $\underline{R}(f,x)$, with repeated application of Lemma~\ref{lem:sound-arithmetic}. 
The induction hypothesis is that up to $l\leq L$ levels of composition of interval operators
produces intervals that contain the result of the corresponding real operators. 
In other words, the real-valued composition $(\psi_l\circ\psi_{l-1}\circ\cdots\circ\psi_1)(\mathbf{x})$ 
is contained in $(\phi_l\circ\phi_{l-1}\circ\cdots\circ\phi_1)(\mathbf{x})$.
The base case comes from the (coordinate-wise) intervals 
$[\mathbf{x},\mathbf{x}], [f(\mathbf{x}),f(\mathbf{x})]$ containing coordinates of $\mathbf{x}$ and $f(\mathbf{x})$; 
the inductive step follows from repeated application of sound floating-point extensions.
\end{proof}

\begin{remark}
While the correction $\tilde{R}-\gamma$, may be a straw argument mitigation for constant $\gamma$, our mitigation can be seen as a sound input-dependent correction.  
To make explicit this required dependence on the input-model pair, we let $\gamma$ be a function of $f, \mathbf{x}$. 
By taking 
$\gamma(f,\mathbf{x}) = \tilde{R}(f,\mathbf{x}) - \underline{R}(f,\mathbf{x})$, we immediately arrive at a sound
certificate via the corresponding correction $\tilde{R}(f,\mathbf{x})-\gamma(f,\mathbf{x}) = \underline{R}(f,\mathbf{x})$ by simply rearranging terms.
\end{remark}

We offer example applications of Theorem~\ref{thm:mitigation} on linear classifiers 
in both weak and strong threat models, with 64-bit floating-point arithmetic.
We use the PyInterval library~\cite{pyinterval} that performs rounded interval arithmetic 
to compute sound $\underline{R}$ and $\overline{\|\boldsymbol{\delta}\|}$ for linear classifiers~\cite{cohen2019certified}.
Our attack success rates for randomly initialized linear classifiers (Section~\ref{sec:exp:linear}) drop to 0\% for all dimensions
in both threat models.

In sum, our theoretical and empirical results provide support for mitigating attacks 
against exact robustness certifications~\cite{cohen2019certified}.
However, we should also note that our certificate $\underline{R}$ is both practically and theoretically safe in the strong threat model.
In the weak threat model, the rounding errors introduced in calculating the norm of perturbation
could potentially invalidate the soundness of $\underline{R}$.
Nontheless, $\underline{R}$ should still mitigate most attacks in practice, as we have shown for linear classifiers.

\section{Discussion}

\subsection{Revisiting attacks}
\label{app:revisit}

We now revisit our Section~\ref{sec:exp:linear} experiment on randomized binary linear classifiers.
Figure~\ref{fig:linear_succ} observes an intriguing behavior of fast initial improvement of 
attack success rate to 50\% followed by asymptoting 
in the weak threat model, for both 32-bit and 64-bit representations.

To explore why our success rate flattens,
we conduct an experiment in the weak threat model with 64-bit representation.
We use the PyInterval library~\cite{pyinterval} to calculate the lower and upper
bounds, $\underline{R}$ and $\overline{R}$, of certified radius $R$.
Then we conduct a binary search in $[\underline{R}, \overline{R}]$
to find the maximum certified radius $\hat{R}$, within which there are no adversarial examples (Theorem~\ref{thm:mitigation}). 
We take $\hat{R}$ as our best approximation to~$R$, 
and use it to estimate the rounding error of~$\tilde{R}$ (\ie $\hat{R}-\tilde{R}$).

As in Figure~\ref{fig:r_error2} of Appendix~\ref{app:res}, our attack against linear models
exploits rounding errors in the magnitude of $10^{-15}$.
We find that the rate that~$\tilde{R}$ is overestimated ($\tilde{R} > \hat{R}$)
flattens around~50\% (Figure~\ref{fig:r_over_under} of Appendix~\ref{app:res}).
Recall that our attack works when~$\tilde{R}$ is overestimated, so there is leeway for our attack to exploit.
Therefore, our success rate flattens around 50\% because that is the maximum rate that $\tilde{R}$ is overestimated.

Why does the overestimated rate flatten around 50\%?
One intuition is that rounding is effectively random:
there is a 50\% chance for $\tilde{R}$ to be overestimated (rounded up),
and another 50\% chance for $\tilde{R}$ to be underestimated (rounded down).
Hence, no matter in 32-bit or 64-bit representation,
our attack success rates flatten around 50\%.

\subsection{Limitations}
\label{app:limit_disclosure}

We identify the following limitations that could be addressed as future work.
First, our attacks against methods based on randomized smoothing did not study the relationship between the attack success rate 
and the soundness probability of these methods due to their probabilistic nature. 
A possible future direction would be to study how the two interact. 
Second, we showed how to adopt our mitigation based on interval arithmetic generally, 
and demonstrate this mitigation for linear models and exact certifications. 
Detailed exploration of embedding these mitigations in other verifiers and models is left as future work.
Third, our evaluations demonstrate the effectiveness of our attacks
on relatively small models (e.g., linear SVM) and datasets (e.g., MNIST).
Further evaluation on state-of-the-art models and datasets is also left as future work.

Our attacks could potentially be used to find violations of robustness guarantees in system implementations,
by exploiting floating-point vulnerabilities in them.
Note that to date, we are not aware of any certified robustness implementations yet in deployment. 
However since certifications could be deployed in the near future, this work serves an important role in highlighting this new vulnerability, 
and in proposing an effective mitigation against it.
\section{Related work}
\label{sec:related}

Several works have explored the influence of floating-point representation on guarantees of verified neural networks.
For example, verifiers designed for floating-point representation
have been shown to not necessarily work for quantized neural networks~\cite{giacobbe2020many,henzinger2021scalable,jia2020efficient}.

The closest to our work is the independent work by~\cite{jia2021exploiting} 
who also exploit rounding errors to {discover violations of network robustness certifications.} %
Our work differs from~\cite{jia2021exploiting} on the adversarial examples we {find}. %
As we show in Section~\ref{sec:crown}, we are able to find an adversarial example $\mathbf{x}'$ 
for unaltered natural image $\mathbf{x}$ from test data, within that image $\mathbf{x}$'s certified radius.
The work by Jia and Rinard, instead, {does not find certification-violating adversarial examples of test instances. 
It finds perturbed inputs~$\mathbf{x}_0'$ of synthetic inputs~$\mathbf{x}_0$, that violate certifications of $\mathbf{x}_0$. 
In particular, they adjust brightness of a natural test image~$\mathbf{x}$ to produce a~$\mathbf{x}_0$.} 
That is, {their attack point}~$\mathbf{x}_0'$ is \emph{outside {the} certified radius of~$\mathbf{x}$.}
Hence, our attack can be seen as a stronger attack that is possible due to a novel attack methodology
based on accurate perturbation directions.
In another parallel and independent work, Voracek and Hein~\cite{voravcek2022sound} also attacked randomized smoothing by exploiting the finite representation of floating points.
However, they also attacked synthetic data and model, while we show attacks on unperturbed networks and test-set images and propose a sound and efficient mitigation based on interval arithmetic.

Research in the area of numerical analysis has proposed approaches to address the limitations of floating-point rounding, 
with a focus on measuring the stability of calculations. 
Proposed approaches include replacing floating-point arithmetic with
interval arithmetic~\cite{jaulin2001interval} or affine arithmetic~\cite{de2004affine}.
Both account for rounding errors and return an interval that contains the correct result.
Singh \etal~\cite{singh2018fast, singh2019abstract} use interval arithmetic 
in neural network certification like what we do for our mitigation, and are sound w.r.t.~floating-point arithmetic.
However, they use the $\ell_{\inf}$ metric in neural network certification,
which avoids the influence of potential rounding errors in the calculation of the perturbation norm.
We use the $\ell_2$ metric, and handle the impact of rounding errors in the calculation of 
the perturbation norm for certification in the strong threat model. Moreover, our mitigation can be used to make 
implementations of any valid certification mechanisms sound, by adopting interval arithmetic due to Theorem~\ref{thm:mitigation}.
We adopt interval arithmetic with the implementation PyInterval~\cite{pyinterval} 
in the calculation of robustness certification.

Finally, exact robustness certifications without rounding errors can be achieved under special settings.
The exact robustness certification problem can be transformed into a 
mixed integer linear programming (MILP) problem~\cite{tjeng2017evaluating},
which can be solved over rational numbers without floating-point errors~\cite{applegate2007exact}.
However to our knowledge, currently there is no work that applies exact MILP solvers to
neural network robustness certification, given their very limited performance.
Jia and Rinard~\cite{jia2020efficient} did exact certification on a special type of neural network --- binarized neural networks (BNNs),
which quantize their weights and activations to be binary. 
Such approaches do not involve floating-point arithmetic, and hence are not vulnerable to floating-point errors.

\section{Conclusion}
Certified robustness has been proposed as a defense against adversarial examples.
In this work we have shown that guarantees of several certification mechanisms
do not hold in practice since they rely on real numbers that are approximated on modern computers.
Hence, computation on floating-point numbers---used to represent real numbers---can overestimate
 certification guarantees due to rounding.
We propose and evaluate a rounding search method that finds adversarial inputs on
linear classifiers and verified neural networks 
within their certified radii---violating their certification guarantees.
We propose rounded interval arithmetic as the mitigation,
by accounting for the rounding errors involved in the computation
of certification guarantees.
We conclude that if certified robustness is to be used for security-critical
applications, their guarantees and implementations need to account for 
limitations of modern computing architecture.

\section{Acknowledgment}

This work was supported by the joint CATCH MURI-AUSMURI,
and The University of Melbourne's Research
Computing Services and the Petascale Campus Initiative.
The first author is supported by the University of Melbourne 
research scholarship (MRS) scheme.

\bibliography{refs}


\begin{thebibliography}{30}


\ifx \showCODEN    \undefined \def \showCODEN     #1{\unskip}     \fi
\ifx \showDOI      \undefined \def \showDOI       #1{#1}\fi
\ifx \showISBNx    \undefined \def \showISBNx     #1{\unskip}     \fi
\ifx \showISBNxiii \undefined \def \showISBNxiii  #1{\unskip}     \fi
\ifx \showISSN     \undefined \def \showISSN      #1{\unskip}     \fi
\ifx \showLCCN     \undefined \def \showLCCN      #1{\unskip}     \fi
\ifx \shownote     \undefined \def \shownote      #1{#1}          \fi
\ifx \showarticletitle \undefined \def \showarticletitle #1{#1}   \fi
\ifx \showURL      \undefined \def \showURL       {\relax}        \fi
\providecommand\bibfield[2]{#2}
\providecommand\bibinfo[2]{#2}
\providecommand\natexlab[1]{#1}
\providecommand\showeprint[2][]{arXiv:#2}

\bibitem[Applegate et~al\mbox{.}(2007)]%
        {applegate2007exact}
\bibfield{author}{\bibinfo{person}{David~L Applegate}, \bibinfo{person}{William Cook}, \bibinfo{person}{Sanjeeb Dash}, {and} \bibinfo{person}{Daniel~G Espinoza}.} \bibinfo{year}{2007}\natexlab{}.
\newblock \showarticletitle{Exact solutions to linear programming problems}.
\newblock \bibinfo{journal}{\emph{Operations Research Letters}} \bibinfo{volume}{35}, \bibinfo{number}{6} (\bibinfo{year}{2007}), \bibinfo{pages}{693--699}.
\newblock


\bibitem[Brown et~al\mbox{.}(2001)]%
        {binomial_prop}
\bibfield{author}{\bibinfo{person}{Lawrence~D. Brown}, \bibinfo{person}{T.~Tony Cai}, {and} \bibinfo{person}{Anirban DasGupta}.} \bibinfo{year}{2001}\natexlab{}.
\newblock \showarticletitle{Interval estimation for a binomial proportion}.
\newblock \bibinfo{journal}{\emph{Statistical science}} \bibinfo{volume}{16}, \bibinfo{number}{2} (\bibinfo{year}{2001}), \bibinfo{pages}{101--133}.
\newblock


\bibitem[Bunel et~al\mbox{.}(2020)]%
        {bunel2020branch}
\bibfield{author}{\bibinfo{person}{Rudy Bunel}, \bibinfo{person}{P Mudigonda}, \bibinfo{person}{Ilker Turkaslan}, \bibinfo{person}{Philip Torr}, \bibinfo{person}{Jingyue Lu}, {and} \bibinfo{person}{Pushmeet Kohli}.} \bibinfo{year}{2020}\natexlab{}.
\newblock \showarticletitle{Branch and bound for piecewise linear neural network verification}.
\newblock \bibinfo{journal}{\emph{Journal of Machine Learning Research}} \bibinfo{volume}{21}, \bibinfo{number}{2020} (\bibinfo{year}{2020}).
\newblock


\bibitem[Carlini and Wagner(2017)]%
        {carlini2017towards}
\bibfield{author}{\bibinfo{person}{Nicholas Carlini} {and} \bibinfo{person}{David Wagner}.} \bibinfo{year}{2017}\natexlab{}.
\newblock \showarticletitle{Towards evaluating the robustness of neural networks}. In \bibinfo{booktitle}{\emph{IEEE Symposium on Security and Privacy (S\&P)}}. \bibinfo{pages}{39--57}.
\newblock


\bibitem[Cohen et~al\mbox{.}(2019)]%
        {cohen2019certified}
\bibfield{author}{\bibinfo{person}{Jeremy Cohen}, \bibinfo{person}{Elan Rosenfeld}, {and} \bibinfo{person}{Zico Kolter}.} \bibinfo{year}{2019}\natexlab{}.
\newblock \showarticletitle{Certified adversarial robustness via randomized smoothing}. In \bibinfo{booktitle}{\emph{International Conference on Machine Learning (ICML)}}. PMLR, \bibinfo{pages}{1310--1320}.
\newblock


\bibitem[De~Figueiredo and Stolfi(2004)]%
        {de2004affine}
\bibfield{author}{\bibinfo{person}{Luiz~Henrique De~Figueiredo} {and} \bibinfo{person}{Jorge Stolfi}.} \bibinfo{year}{2004}\natexlab{}.
\newblock \showarticletitle{Affine arithmetic: concepts and applications}.
\newblock \bibinfo{journal}{\emph{Numerical Algorithms}} \bibinfo{volume}{37}, \bibinfo{number}{1} (\bibinfo{year}{2004}), \bibinfo{pages}{147--158}.
\newblock


\bibitem[Ehlers(2017)]%
        {ehlers2017formal}
\bibfield{author}{\bibinfo{person}{Ruediger Ehlers}.} \bibinfo{year}{2017}\natexlab{}.
\newblock \showarticletitle{Formal verification of piece-wise linear feed-forward neural networks}. In \bibinfo{booktitle}{\emph{Automated Technology for Verification and Analysis: 15th International Symposium, ATVA 2017, Pune, India, October 3--6, 2017, Proceedings 15}}. Springer, \bibinfo{pages}{269--286}.
\newblock


\bibitem[Giacobbe et~al\mbox{.}(2020)]%
        {giacobbe2020many}
\bibfield{author}{\bibinfo{person}{Mirco Giacobbe}, \bibinfo{person}{Thomas~A Henzinger}, {and} \bibinfo{person}{Mathias Lechner}.} \bibinfo{year}{2020}\natexlab{}.
\newblock \showarticletitle{How many bits does it take to quantize your neural network?}. In \bibinfo{booktitle}{\emph{International Conference on Tools and Algorithms for the Construction and Analysis of Systems}}. Springer, \bibinfo{pages}{79--97}.
\newblock


\bibitem[Goodfellow et~al\mbox{.}(2015)]%
        {43405}
\bibfield{author}{\bibinfo{person}{Ian Goodfellow}, \bibinfo{person}{Jonathon Shlens}, {and} \bibinfo{person}{Christian Szegedy}.} \bibinfo{year}{2015}\natexlab{}.
\newblock \showarticletitle{Explaining and Harnessing Adversarial Examples}. In \bibinfo{booktitle}{\emph{International Conference on Learning Representations (ICLR)}}.
\newblock
\urldef\tempurl%
\url{http://arxiv.org/abs/1412.6572}
\showURL{%
\tempurl}


\bibitem[{Gurobi Optimization, LLC}(2022)]%
        {gurobi}
\bibfield{author}{\bibinfo{person}{{Gurobi Optimization, LLC}}.} \bibinfo{year}{2022}\natexlab{}.
\newblock \bibinfo{title}{{Gurobi Optimizer Reference Manual}}.
\newblock
\newblock
\urldef\tempurl%
\url{https://www.gurobi.com}
\showURL{%
\tempurl}


\bibitem[Henzinger et~al\mbox{.}(2021)]%
        {henzinger2021scalable}
\bibfield{author}{\bibinfo{person}{Thomas~A Henzinger}, \bibinfo{person}{Mathias Lechner}, {and} \bibinfo{person}{{\DJ}or{\dj}e {\v{Z}}ikeli{\'c}}.} \bibinfo{year}{2021}\natexlab{}.
\newblock \showarticletitle{Scalable verification of quantized neural networks}. In \bibinfo{booktitle}{\emph{Proceedings of the AAAI Conference on Artificial Intelligence (AAAI)}}, Vol.~\bibinfo{volume}{35}. \bibinfo{pages}{3787--3795}.
\newblock


\bibitem[Higham(2002)]%
        {higham2002accuracy}
\bibfield{author}{\bibinfo{person}{Nicholas~J. Higham}.} \bibinfo{year}{2002}\natexlab{}.
\newblock \bibinfo{booktitle}{\emph{Accuracy and stability of numerical algorithms}}.
\newblock \bibinfo{publisher}{Society for Industrial and Applied Mathematics (SIAM)}.
\newblock


\bibitem[Hornik(1991)]%
        {hornik1991approximation}
\bibfield{author}{\bibinfo{person}{Kurt Hornik}.} \bibinfo{year}{1991}\natexlab{}.
\newblock \showarticletitle{Approximation capabilities of multilayer feedforward networks}.
\newblock \bibinfo{journal}{\emph{Neural Networks}} \bibinfo{volume}{4}, \bibinfo{number}{2} (\bibinfo{year}{1991}), \bibinfo{pages}{251--257}.
\newblock


\bibitem[IEEE({[n.\,d.]})]%
        {IEEE-FP}
\bibfield{author}{\bibinfo{person}{IEEE}.} \bibinfo{year}{[n.\,d.]}\natexlab{}.
\newblock \showarticletitle{{IEEE} Standard for Floating-Point Arithmetic}.
\newblock \bibinfo{journal}{\emph{IEEE Std 754-2019 (Revision of IEEE 754-2008)}} (\bibinfo{year}{[n.\,d.]}), \bibinfo{pages}{1--84}.
\newblock
\urldef\tempurl%
\url{https://doi.org/10.1109/IEEESTD.2019.8766229}
\showDOI{\tempurl}


\bibitem[Jaulin et~al\mbox{.}(2001)]%
        {jaulin2001interval}
\bibfield{author}{\bibinfo{person}{Luc Jaulin}, \bibinfo{person}{Michel Kieffer}, \bibinfo{person}{Olivier Didrit}, {and} \bibinfo{person}{Eric Walter}.} \bibinfo{year}{2001}\natexlab{}.
\newblock \showarticletitle{Interval analysis}.
\newblock In \bibinfo{booktitle}{\emph{Applied Interval Analysis}}. \bibinfo{publisher}{Springer}, \bibinfo{pages}{11--43}.
\newblock


\bibitem[Jia and Rinard(2020)]%
        {jia2020efficient}
\bibfield{author}{\bibinfo{person}{Kai Jia} {and} \bibinfo{person}{Martin Rinard}.} \bibinfo{year}{2020}\natexlab{}.
\newblock \showarticletitle{Efficient exact verification of binarized neural networks}.
\newblock \bibinfo{journal}{\emph{Advances in neural information processing systems}}  \bibinfo{volume}{33} (\bibinfo{year}{2020}), \bibinfo{pages}{1782--1795}.
\newblock


\bibitem[Jia and Rinard(2021)]%
        {jia2021exploiting}
\bibfield{author}{\bibinfo{person}{Kai Jia} {and} \bibinfo{person}{Martin Rinard}.} \bibinfo{year}{2021}\natexlab{}.
\newblock \showarticletitle{Exploiting verified neural networks via floating point numerical error}. In \bibinfo{booktitle}{\emph{International Static Analysis Symposium}}. Springer, \bibinfo{pages}{191--205}.
\newblock


\bibitem[LeCun et~al\mbox{.}(2010)]%
        {mnist}
\bibfield{author}{\bibinfo{person}{Yann LeCun}, \bibinfo{person}{Corinna Cortes}, {and} \bibinfo{person}{Christopher~J.C. Burges}.} \bibinfo{year}{2010}\natexlab{}.
\newblock \showarticletitle{{MNIST} handwritten digit database}.
\newblock \bibinfo{journal}{\emph{ATT Labs [Online]. Available: http://yann.lecun.com/exdb/mnist}}  \bibinfo{volume}{2} (\bibinfo{year}{2010}).
\newblock


\bibitem[Lecuyer et~al\mbox{.}(2019)]%
        {augmentation}
\bibfield{author}{\bibinfo{person}{Mathias Lecuyer}, \bibinfo{person}{Vaggelis Atlidakis}, \bibinfo{person}{Roxana Geambasu}, \bibinfo{person}{Daniel Hsu}, {and} \bibinfo{person}{Suman Jana}.} \bibinfo{year}{2019}\natexlab{}.
\newblock \showarticletitle{Certified robustness to adversarial examples with differential privacy}. In \bibinfo{booktitle}{\emph{IEEE Symposium on Security and Privacy (S\&P)}}. \bibinfo{pages}{656--672}.
\newblock


\bibitem[Li et~al\mbox{.}(2019)]%
        {li2019certified}
\bibfield{author}{\bibinfo{person}{Bai Li}, \bibinfo{person}{Changyou Chen}, \bibinfo{person}{Wenlin Wang}, {and} \bibinfo{person}{Lawrence Carin}.} \bibinfo{year}{2019}\natexlab{}.
\newblock \showarticletitle{Certified adversarial robustness with additive noise}. In \bibinfo{booktitle}{\emph{Advances in Neural Information Processing Systems (NeurIPS)}}, Vol.~\bibinfo{volume}{32}.
\newblock


\bibitem[Madry et~al\mbox{.}(2018)]%
        {madry2017towards}
\bibfield{author}{\bibinfo{person}{Aleksander Madry}, \bibinfo{person}{Aleksandar Makelov}, \bibinfo{person}{Ludwig Schmidt}, \bibinfo{person}{Dimitris Tsipras}, {and} \bibinfo{person}{Adrian Vladu}.} \bibinfo{year}{2018}\natexlab{}.
\newblock \showarticletitle{Towards deep learning models resistant to adversarial attacks}. In \bibinfo{booktitle}{\emph{International Conference on Learning Representations (ICLR)}}.
\newblock


\bibitem[Singh et~al\mbox{.}(2018)]%
        {singh2018fast}
\bibfield{author}{\bibinfo{person}{Gagandeep Singh}, \bibinfo{person}{Timon Gehr}, \bibinfo{person}{Matthew Mirman}, \bibinfo{person}{Markus P{\"u}schel}, {and} \bibinfo{person}{Martin Vechev}.} \bibinfo{year}{2018}\natexlab{}.
\newblock \showarticletitle{Fast and effective robustness certification}.
\newblock \bibinfo{journal}{\emph{Advances in neural information processing systems}}  \bibinfo{volume}{31} (\bibinfo{year}{2018}).
\newblock


\bibitem[Singh et~al\mbox{.}(2019)]%
        {singh2019abstract}
\bibfield{author}{\bibinfo{person}{Gagandeep Singh}, \bibinfo{person}{Timon Gehr}, \bibinfo{person}{Markus P{\"u}schel}, {and} \bibinfo{person}{Martin Vechev}.} \bibinfo{year}{2019}\natexlab{}.
\newblock \showarticletitle{An abstract domain for certifying neural networks}.
\newblock \bibinfo{journal}{\emph{Proceedings of the ACM on Programming Languages}} \bibinfo{volume}{3}, \bibinfo{number}{POPL} (\bibinfo{year}{2019}), \bibinfo{pages}{1--30}.
\newblock


\bibitem[Szegedy et~al\mbox{.}(2014)]%
        {42503}
\bibfield{author}{\bibinfo{person}{Christian Szegedy}, \bibinfo{person}{Wojciech Zaremba}, \bibinfo{person}{Ilya Sutskever}, \bibinfo{person}{Joan Bruna}, \bibinfo{person}{Dumitru Erhan}, \bibinfo{person}{Ian Goodfellow}, {and} \bibinfo{person}{Rob Fergus}.} \bibinfo{year}{2014}\natexlab{}.
\newblock \showarticletitle{Intriguing properties of neural networks}. In \bibinfo{booktitle}{\emph{International Conference on Learning Representations (ICLR)}}.
\newblock
\urldef\tempurl%
\url{http://arxiv.org/abs/1312.6199}
\showURL{%
\tempurl}


\bibitem[Taschini(2008)]%
        {pyinterval}
\bibfield{author}{\bibinfo{person}{Stefano Taschini}.} \bibinfo{year}{2008}\natexlab{}.
\newblock \bibinfo{title}{{PyInterval}, interval arithmetic in {P}ython}.
\newblock
\newblock
\urldef\tempurl%
\url{https://pypi.org/project/pyinterval/}
\showURL{%
\tempurl}
\newblock
\shownote{version 1.2.0 released 2017-03-05}.


\bibitem[Tjeng et~al\mbox{.}(2019)]%
        {tjeng2017evaluating}
\bibfield{author}{\bibinfo{person}{Vincent Tjeng}, \bibinfo{person}{Kai~Y. Xiao}, {and} \bibinfo{person}{Russ Tedrake}.} \bibinfo{year}{2019}\natexlab{}.
\newblock \showarticletitle{Evaluating Robustness of Neural Networks with Mixed Integer Programming}. In \bibinfo{booktitle}{\emph{International Conference on Learning Representations}}.
\newblock
\urldef\tempurl%
\url{https://openreview.net/forum?id=HyGIdiRqtm}
\showURL{%
\tempurl}


\bibitem[Vor{\'a}{\v{c}}ek and Hein(2022)]%
        {voravcek2022sound}
\bibfield{author}{\bibinfo{person}{V{\'a}clav Vor{\'a}{\v{c}}ek} {and} \bibinfo{person}{Matthias Hein}.} \bibinfo{year}{2022}\natexlab{}.
\newblock \showarticletitle{Sound randomized smoothing in floating-point arithmetics}.
\newblock \bibinfo{journal}{\emph{arXiv preprint arXiv:2207.07209}} (\bibinfo{year}{2022}).
\newblock


\bibitem[Wang et~al\mbox{.}(2021)]%
        {beta-crown}
\bibfield{author}{\bibinfo{person}{Shiqi Wang}, \bibinfo{person}{Huan Zhang}, \bibinfo{person}{Kaidi Xu}, \bibinfo{person}{Xue Lin}, \bibinfo{person}{Suman Jana}, \bibinfo{person}{Cho-Jui Hsieh}, {and} \bibinfo{person}{J~Zico Kolter}.} \bibinfo{year}{2021}\natexlab{}.
\newblock \showarticletitle{Beta-{CROWN}: Efficient bound propagation with per-neuron split constraints for complete and incomplete neural network verification}. In \bibinfo{booktitle}{\emph{Advances in Neural Information Processing Systems (NeurIPS)}}.
\newblock


\bibitem[Zhang et~al\mbox{.}(2018)]%
        {crown}
\bibfield{author}{\bibinfo{person}{Huan Zhang}, \bibinfo{person}{Tsui-Wei Weng}, \bibinfo{person}{Pin-Yu Chen}, \bibinfo{person}{Cho-Jui Hsieh}, {and} \bibinfo{person}{Luca Daniel}.} \bibinfo{year}{2018}\natexlab{}.
\newblock \showarticletitle{Efficient neural network robustness certification with general activation functions}. In \bibinfo{booktitle}{\emph{Advances in Neural Information Processing Systems (NeurIPS)}}, Vol.~\bibinfo{volume}{31}.
\newblock


\bibitem[Zombori et~al\mbox{.}(2020)]%
        {zombori2020fooling}
\bibfield{author}{\bibinfo{person}{D{\'a}niel Zombori}, \bibinfo{person}{Bal{\'a}zs B{\'a}nhelyi}, \bibinfo{person}{Tibor Csendes}, \bibinfo{person}{Istv{\'a}n Megyeri}, {and} \bibinfo{person}{M{\'a}rk Jelasity}.} \bibinfo{year}{2020}\natexlab{}.
\newblock \showarticletitle{Fooling a complete neural network verifier}. In \bibinfo{booktitle}{\emph{International Conference on Learning Representations (ICLR)}}.
\newblock


\end{thebibliography}
\bibliographystyle{ACM-Reference-Format}

\appendix

\section{Attack results}
\label{app:res}

The results of rounding search attacks on linear SVM from Section~\ref{sec:lrmodel} 
for each combination of the {distinct} labels $i,j\in\{0,\ldots,9\}$ of 
the MNIST dataset are listed in Table~\ref{tab:svm_rates} 
for the weak model and the strong model.
An example of original and adversarial images together with their perturbation and certified radius information 
based on linear SVM attack is presented in Figure~\ref{fig:lr_mnist_demo}.

\begin{table}[h]
\caption{
The success rates of our rounding search attack against linear SVM models on 
the MNIST dataset in the weak and strong threat models (Section~\ref{sec:lrmodel}).
In the weak model,
an experiment is successful if adversarial perturbation~$\|\boldsymbol{\delta}\|$ 
is less than or equal to certified radius $\tilde{R} = |\mathbf{w}^T\mathbf{x}+b|/\|\mathbf{w}\|$
computed with finite-precision floating-point arithmetic (\ie $\|\boldsymbol{\delta}\| \le \tilde{R}$).
In the strong model, an experiment is successful if the upper bound of adversarial perturbation~$\overline{\|\boldsymbol{\delta}\|}$ 
is less than or equal to certified radius $\tilde{R}$.
45 linear SVM models have been trained and attacked for each combination of labels respectively.
A cell in row $i$ and column $j$ reports the attack success rate for classes original $i$ and target $j$, 
plus the attack success rate for classes original $j$ and target $i$. 
For example, in the weak model, for an SVM model, we can find adversarial perturbations within 
certified radius for 11.06\% of all images labelled 1 (with the model classifying them as 0) 
or images labelled 0 (with the model classifying them as 1).
We use ``-'' to denote models where rounding search did not find an adversarial example. 
}\label{tab:svm_rates}
    \begin{subtable}{\linewidth}\centering
    {
    \addtolength{\tabcolsep}{-0.4em}
    \renewcommand{\arraystretch}{1.15}
    \begin{tabular}{cccccccccc}
    \hline
    labels & 1       & 2      & 3      & 4      & 5       & 6      & 7       & 8      & 9      \\ \hline
    0      & 11.06\% & 5.12\% & 3.62\% & 2.65\% & 0.53\%  & 2.58\% & 6.18\%  & 3.43\% & 3.12\% \\
    1      &         & 5.77\% & 5.64\% & 6.8\%  & 23.24\% & 0.38\% & 16.04\% & 4.46\% & 6.06\% \\
    2      &         &        & 1.08\% & 0.3\%  & 1.3\%   & -      & 1.12\%  & 1.0\%  & 1.03\% \\
    3      &         &        &        & 0.1\%  & 1.21\%  & 5.18\% & 4.47\%  & 1.31\% & 0.25\% \\
    4      &         &        &        &        & 1.92\%  & 0.05\% & 4.08\%  & 1.84\% & 0.4\%  \\
    5      &         &        &        &        &         & 2.16\% & 5.78\%  & 0.05\% & 0.58\% \\
    6      &         &        &        &        &         &        & 2.92\%  & 10.09\% & 0.05\% \\
    7      &         &        &        &        &         &        &         & 2.3\%  & 1.57\% \\
    8      &         &        &        &        &         &        &         &        & 1.06\% \\ \hline
    \end{tabular}
    }
    \caption{weak model}\label{tab:svm_mnist_res_weak}
    \end{subtable}
    \\
    \begin{subtable}{\linewidth}\centering
    {
    \addtolength{\tabcolsep}{-0.1em}
    \renewcommand{\arraystretch}{1.15}
    \begin{tabular}{cccccccccc}
    \hline
    labels & 1       & 2      & 3      & 4      & 5       & 6      & 7       & 8      & 9      \\ \hline
    0      & -       & -      & -      & -      & -       & -      & -       & -      & 0.05\% \\
    1      &         & -      & 0.05\% & 0.05\% & -       & -      & -       & -      & -      \\
    2      &         &        & -      & -      & -       & -      & -       & -      & -      \\
    3      &         &        &        & -      & 0.16\%  & 0.05\% & -       & 0.05\% & -      \\
    4      &         &        &        &        & -       & 0.05\% & -       & 0.05\% & -      \\
    5      &         &        &        &        &         & -      & -       & 0.05\% & 0.16\% \\
    6      &         &        &        &        &         &        & -       & -      & 0.05\% \\
    7      &         &        &        &        &         &        &         & -      & -      \\
    8      &         &        &        &        &         &        &         &        & -      \\ \hline
    \end{tabular}
    }
    \caption{strong model}\label{tab:svm_mnist_res_strong}
    \end{subtable}
\end{table}

\begin{figure}[h]
\centering
\begin{subfigure}{0.25\textwidth}
\centering
    \includegraphics[width=.9\linewidth]{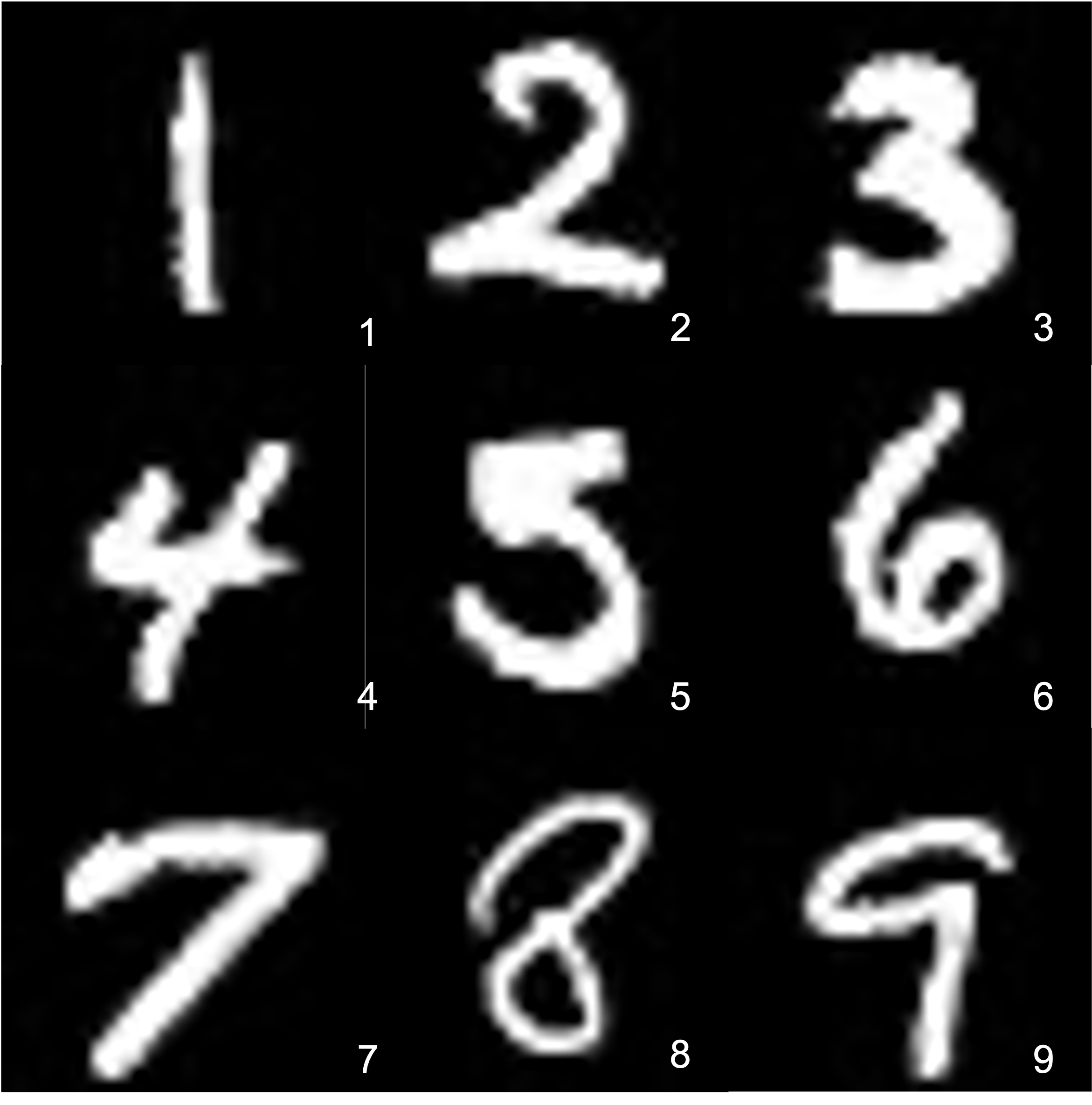}
    \caption{}
\end{subfigure}%
\begin{subfigure}{0.25\textwidth}
\centering
    \includegraphics[width=.9\linewidth]{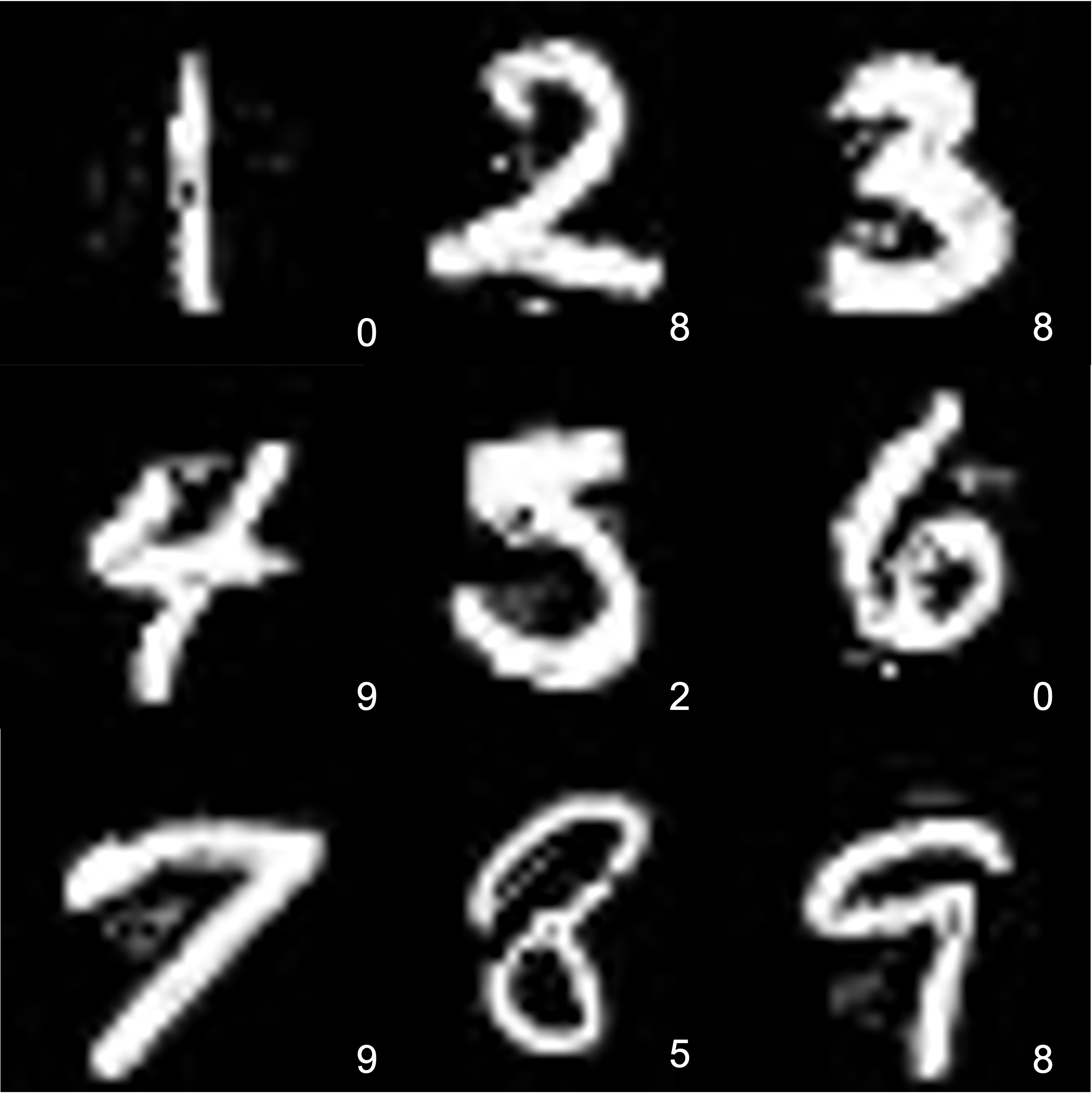}
    \caption{}
\end{subfigure}
\caption{
(a) Original images from the MNIST dataset.
(b) Corresponding adversarial images in the weak threat model,
with perturbations within the exact certified radius (\ie {$\|\boldsymbol{\delta}\| \le \tilde{R}$}) 
but the linear SVM model misclassifies them.
For example, the first {top-left} image in (a) has certified radius of {$\tilde{R}=333.608776491892\emph{5}$},
and is classified as 1, while the corresponding adversarial image in (b) has perturbation 
$\|\boldsymbol{\delta}\|=333.608776491892\emph{4}$, and is classified as 0.
Labels at the bottom right of each image are the classifications of the linear SVM model (Section~\ref{sec:lrmodel}).
}
\label{fig:lr_mnist_demo}
\end{figure}

The results of rounding search attacks on neural nets with approximate certification 
(\ie randomized smoothing) are listed in Table~\ref{tab:rs_res}.

\begin{table}[H]
  \centering
  \caption{
    The attack success rates of our rounding search within the certified radius $\tilde{R}$~($\alpha=0.1\%)$~\cite{cohen2019certified}. 
    The number of Monte Carlo {samples} $M$ used in certified radius estimation ranges in $\{100, 1000, 10000\}$.
    Prediction noise scale $\sigma_P\in\{1.0, 3.0, 5.0, 7.0\}$.
    We randomly sample and evaluate $N=1000$ neighbors $\boldsymbol{\delta}'$ of $\boldsymbol{\delta}$ in the search area for each image.
    $N_A$ and $N_V$ denote the number of examples that have been successfully attacked and verified respectively.
  }
  \renewcommand{\arraystretch}{1.15}
  \begin{tabular}{ccccc}
    \hline
    $\sigma_P$           & $M$ & $N_A$ & $N_V$ & success rate \\ \hline
    \multirow{3}{*}{1.0} & 100   & 5     & 2111  & 0.24\%       \\
                          & 1000  & 1     & 2113  & 0.05\%       \\
                          & 10000 & 1     & 2115  & 0.05\%       \\ \hline
    \multirow{3}{*}{3.0} & 100   & 182   & 2066  & 8.81\%       \\
                          & 1000  & 89    & 2111  & 4.22\%       \\
                          & 10000 & 8     & 2113  & 0.38\%       \\ \hline
    \multirow{3}{*}{5.0} & 100   & 273   & 1606  & 17.00\%      \\
                          & 1000  & 162   & 2097  & 7.73\%       \\
                          & 10000 & 54    & 2109  & 2.56\%       \\ \hline
    \multirow{3}{*}{7.0} & 100   & 220   & 1042  & 21.11\%      \\
                          & 1000  & 248   & 2026  & 12.24\%      \\
                          & 10000 & 99    & 2066  & 4.79\%       \\ \hline
    \end{tabular}
  \label{tab:rs_res}
\end{table}

The results of exploration of flattening success rate phenomenon 
appear in Figure~\ref{fig:random_analysis}.

\begin{figure}[H]
\centering
\begin{subfigure}{0.5\textwidth}
\centering
    \includegraphics[width=0.8\linewidth]{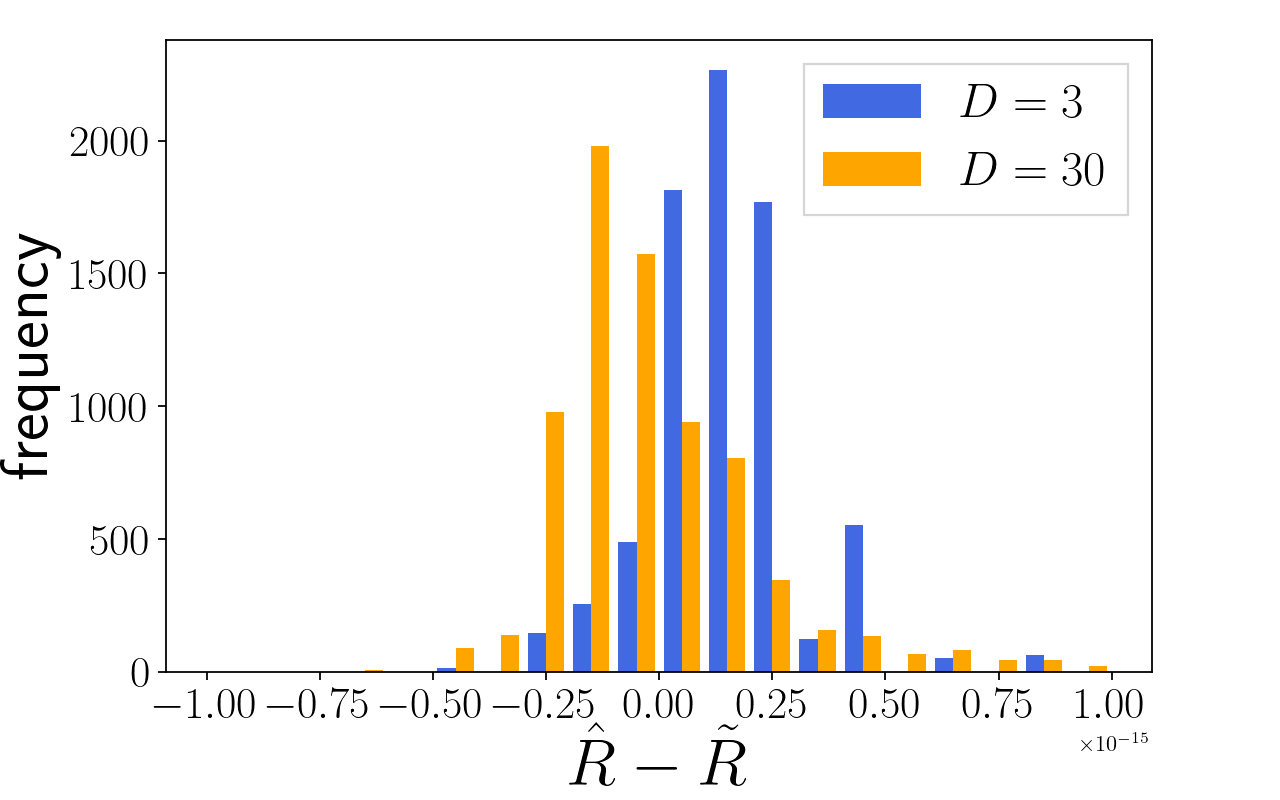}
    \caption{}
    \label{fig:r_error2}
\end{subfigure} \\
\begin{subfigure}{0.5\textwidth}
\centering
    \includegraphics[width=0.8\linewidth]{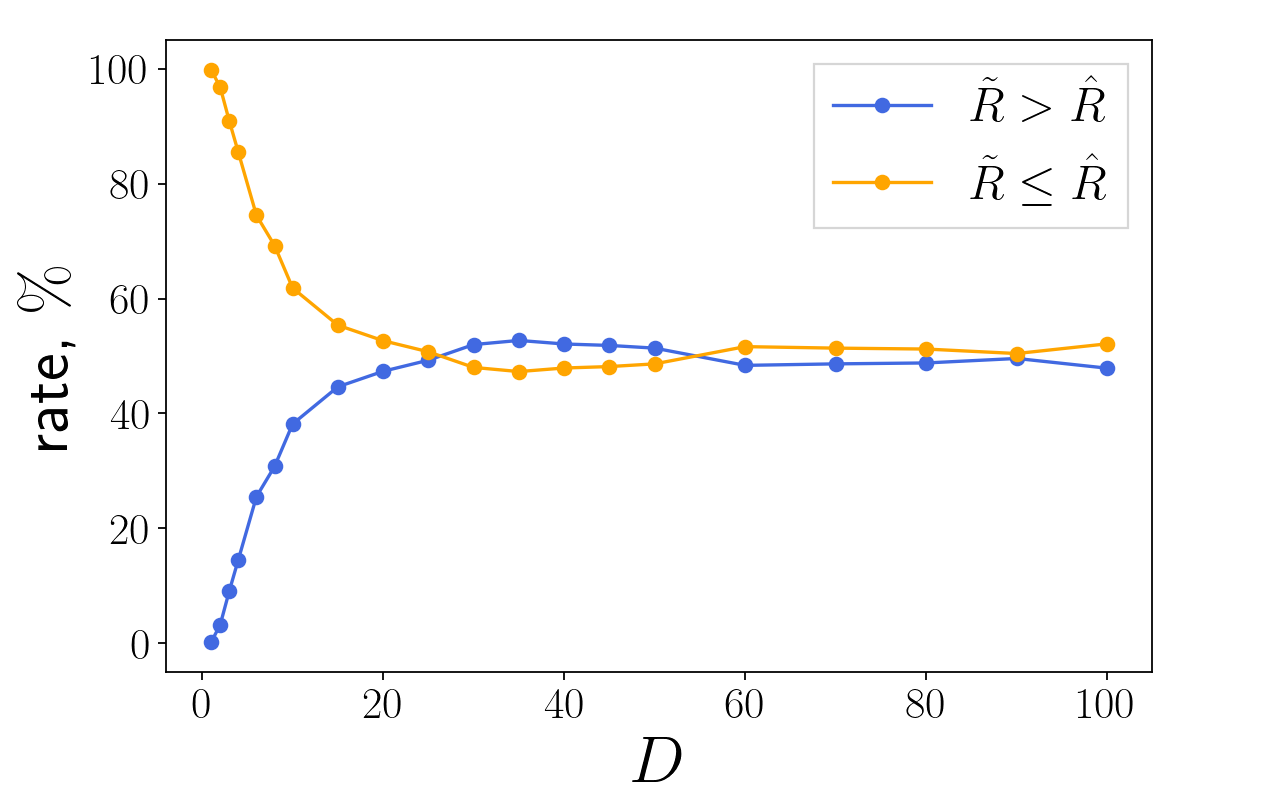}
    \caption{}
    \label{fig:r_over_under}
\end{subfigure}
\caption{
Exploration of flattening success rate phenomenon in Section~\ref{sec:exp:linear}.
The experiment is done in the weak threat model with 64-bit floating-point representation.
$\tilde{R}$ is the estimated certified radius using IEEE 754 floating-point arithmetic (with rounding errors),
$\hat{R}$ is the certified radius found via binary search within $[\underline{R}, \overline{R}]$,
$\underline{R}$ and $\overline{R}$ are the lower and upper bounds of the
real certified radius $R$ estimated using interval arithmetic~\cite{pyinterval}.
$\hat{R}$ is our best approximation to $R$, within which there are no robustness violations.
(a) Deviations ($\hat{R}-\tilde{R}$) between $\tilde{R}$ and $\hat{R}$ 
over 10 000 trials for $D\in\{3,30\}$.
(b) For each dimension $D\in[1,100]$, we plot percentage of 10 000 trials for which $\hat{R}>\tilde{R}$ and $\hat{R}\le\tilde{R}$.
Our attacks may work when $\hat{R}<\tilde{R}$, that is, the certified radius is overestimated.
}
\label{fig:random_analysis}
\end{figure}

\end{document}